\documentclass[journal,12pt,onecolumn]{IEEEtran}
\usepackage[dvips]{graphicx}

\usepackage{amssymb}
\usepackage{amsmath}
\usepackage{color}
\usepackage{ booktabs}
\usepackage{array}  
\usepackage{research5IEEE}

\usepackage{amsfonts}
\usepackage{amssymb}
\usepackage{array}

\usepackage{amsthm}

\newtheorem{Theorem}{Theorem}
\newtheorem{Remark}{Remark}

\newtheorem{Corollary}{Corollary}

\ifCLASSINFOpdf

\else

\fi

\begin{document}
\title{Achievable Rates for Discrete Memoryless Multicast Networks With and Without Feedback}
\author{
  \IEEEauthorblockN{Youlong Wu\\}
\thanks{This paper was in part presented at the \emph{IEEE Information Theory Workshop}, in  Jeju, Korea, Oct 2015, and at the \emph{53rd Annual Allerton Conference on Communication, Control, and Computing }, in  Monticello, IL, USA, Sep 2015. }

}
\maketitle

 \begin{abstract}
Coding schemes for  discrete memoryless multicast networks (DM-MN) with rate-limited  feedback from the receivers and relays to the transmitter are proposed.  The schemes improve over the noisy network coding proposed by Lim \emph{et al.}.   For the single relay channel with relay-transmitter feedback, our coding schemes  recover Gabbai and Bross's results, and strictly improve on  noisy network coding,  distributed decode-forward coding proposed by Lim \emph{et al.},  and   all  known lower bounds on the achievable rates proposed for the setup without feedback. 

The coding schemes are based on block-Markov coding, superposition coding, sliding-window/backward decoding and hybrid relaying strategies.  In our Scheme 1A,   the relays and receivers use compress-forward strategy   and send the compression indices to the transmitter through the feedback pipes.  After obtaining the compression indices through feedback, the transmitter sends them together with the source message. Each receiver uses backward decoding to jointly decode the source message and all compression indices. Our Scheme 1B is similar, except that  here each relay not only uses compress forward to compress its channel outputs, but also performs partial decode forward to decode a common part of source message.  Our Scheme 1C allows different relays to decode different parts of the source message (in Scheme 1A, no relay decodes any part of the source message; in Scheme 1B, all relays decode the same part of source message), which can  achieve higher rates than Scheme 1A and Scheme 1B.

Motivated by the feedback coding scheme, we propose  a new coding scheme for DM-MN \emph{without}  feedback, where the transmitter decodes the compression indices instead of obtaining them directly through feedback. It is shown that the scheme still  improves  noisy network coding and  distributed decode-forward coding for some channels.

 \end{abstract}

\section{Introduction}
The relay channel \cite{Meulen'74}  describes a 3-node communication channel where the transmitter sends a message to the receiver with the assistance of a relay.  Cover and El Gamal \cite{Cover'79} proposed two basic coding strategies: compress-forward and decode-forward that are based on block-Markov coding. The  compress-forward strategy has the relay compress its outputs and send the compression index to the receiver.  The  decode-forward strategy has the relay first decode all or part of the message  and then send the decoded message to the receiver. Both strategies have been generalized to   multiple-relay channels in \cite{Kramer'05,Xie'05, Yassaee'11, Yassaee'08,Li'15}. The compress-forward strategy was later extended to  multi-message multicast and  multi-messages networks, the so called \emph{noisy network coding} (NNC) \cite{Lim'11, Hou'13}. Recently, a \emph{distributed decode-forward coding} (DDF) scheme was proposed  for   multicast  \cite{Lim'14ITW} and   broadcast relay networks \cite{Lim'14ISIT}, which  uses  the partial decode-forward strategy at the relays and backward encoding at the transmitter.  For $N$-node Gaussian relay networks, both NNC and DDF achieve within constant gap from the capacity. 

Decode-forward and compress-forward require sophisticated operations. A much simpler strategy, called amplify-forward, was introduced by Schein and Gallager \cite{Schein'00} for the 4-node Gaussian diamond network.  When using amplify-forward, the relay scales its received signal and forwards it to the receiver. The  amplify-forward stragegy was generalized to  multihop relay networks in \cite{Agnihotri'11,Maric'12,Niesen'13}. A hybrid coding scheme   that  unifies both amplify-forward and NNC for  general noisy relay networks was proposed in \cite{Minero'11}.

 A different line of works  concentrated on the relay channel with feedback \cite{Cover'79, Gabbai'06}. In \cite{Cover'79} it was shown that perfect feedback from the receiver to the relay makes the relay channel   \emph{physically degraded}, and therefore decode-forward achieves   capacity. For the case of  feedback from the receiver or relay to the transmitter, the capacity is  unknown in general.  In \cite{Gabbai'06} Gabbai and Bross studied this problem and  proposed  inner bounds by using   restricted decoding and deterministic partitioning \cite{Willems'83}. It was shown that  feedback  can strictly improve the no-feedback inner bounds achieved by the compress-forward and decode-forward strategies for some relay channel, for example, the general Gaussian relay channel and $Z$ relay channel. 
 
 In this paper, we  consider the general discrete memoryless multicast network (DM-MN) with   feedback. This network consists of  $N\geq 3$ nodes  where the transmitter sends a source message to different  receivers with the assistance of multiple relays, and each receiver or relay can send feedback signals through a noiseless but rate-limited feedback pipe to the transmitter.  We propose new coding schemes  based on block-Markov coding, superposition coding, {sliding-window/backward decoding} and hybrid relaying strategies. Specifically, in our Scheme 1A,   the relays and receivers use compress-forward strategy   and send the compression indices to the transmitter through the feedback pipes.  After obtaining the compression indices through feedback, the transmitter sends them together with the source message. Each receiver uses backward decoding to jointly decode the source message and all compression indices. Our Scheme 1B is similar, except that  here each relay not only uses compress forward to compress its channel outputs, but also performs partial decode forward to decode a common part of source message.  Our Scheme 1C allows different relays to decode different parts of the source message (in Scheme 1A, no relay decodes any part of the source message; in Scheme 1B, all relays decode the same part of source message), which can  achieve higher rates than Scheme 1A and Scheme 1B.
 
 Our coding schemes (Scheme 1A-1C) are reminiscent of the NNC scheme for general networks \cite{Lim'11, Hou'13} in the sense that the relays and receivers compress their channel outputs and broadcast the compression messages. However,  we introduce combined compress-forward and partial decode-forward strategy  into the relay networks. Moreover,  our  schemes have the transmitter \emph{forward} the receivers' and relays' compression messages, instead of creating a new compression message. This is similar to the schemes proposed in  \cite{Youlong'14} for the broadcast channel with feedback, where the transmitter forwards the receivers' compression messages. Finally, since the transmitter knows the source message and can obtain the compression messages through feedback,  we are able to superpose the transmitter's input on the receivers' and relays' inputs, establishing cooperation between the transmitter and the receivers\&relays.  
 It is shown that our coding schemes generalize Gabbai and Bross's results \cite{Gabbai'06} for the relay channel with relay-transmitter feedback.  For some channels, such as the  Gaussian relay channel and $Z$ relay channels,  our coding schemes  strictly improve over the NNC,  the DDF, and   all  known lower bounds on the achievable rate in the absence of feedback.  
  
  
   Motivated by our feedback coding schemes, we  propose  a new  scheme for  DM-MN \emph{without}  feedback. The key idea is that in each block, instead of obtaining compression messages directly through the feedback pipes,   the transmitter  \emph{decodes} the compression messages  based on its observed channel outputs.  Note that in absence of feedback, the transmitter's input cannot be superposed on  the receivers' and relays' inputs like the feedback case. This is because 
   at the beginning of each block $b\in[1:B]$, the transmitter can only recover the relays' and receivers' inputs of block $b-1$. To ensure the cooperation between the transmitter and the receiver\&relays, in each block $b$, we let the relays and receivers resends some messages that were sent in block $b-1$, which introduces dependence between the inputs of the transmitter and  receivers\&relays. It is shown that our non-feedback coding scheme strictly improves the  NNC and DDF lower bound for some channels.

This paper is organized as follows.  Section II describes the system model:  multicast networks with and without  feedback. Section III gives our main results and  Section IV presents three examples  comparing our lower bounds with  known lower bounds. Section V contains the proofs of  Theorems \ref{Them:unicast1}, \ref{Them:unicast2} and \ref{Them:unicast2.1}.  Finally, Sections  VI gives the proof of  Theorem \ref{Them:unicast3}.

Notation: We use capital letters to denote  random variables and small letters for their realizations, e.g. $X$ and $x$. For nonnegative integers $k,j$,  let $X^j_k:= (X_k,\ldots,X_j)$ and $x^j_k:= (x_{k,1},\ldots,x_{k,j})$. 

For a  set of integers $\set{A}\subseteq[1:N]$, we denote by $|\set{A}|$ its cardinality and its complement  by $\set{A}^c:=[1:N]\backslash \set{A}$. A tuple of random variables is denoted as $X(\set{A}):=[X_k:k\in\set{A}]$. 


Given a distribution $P_A$ over some alphabet $\set{A}$, a positive real number $\epsilon>0$, and a positive integer $n$, $\set{T}_{\epsilon}^{(n)}(P_A)$ is the typical set in \cite{Cover'book}.  Given a positive integer $n$, let $\mathbf{1}_{[n]}$ denote the all-one tuple of length $n$, e.g. $\mathbf{1}_{[3]}=(1,1,1)$. Define a function $\mathcal{C}(x):=\frac{1}{2}\log_2(1+x)$. 

\section{System model}
\subsection{Discrete  memoryless multicast networks}\label{Model:nofb}
Consider an $N$-node DM-MN, see Fig. \ref{fig:Mrelays}. Let Node 1 be the transmitter, and $\set{R}$ and $\set{D}$ denotes the set of relays and receivers, respectively, 
where $\set{R}\subset[2:N]$ and $\set{D}=[2:N]\backslash \set{R}$.  This setup is characterized by $2N$  finite alphabets $\set{X}_1,\ldots,\set{X}_N, \set{Y}_1,\ldots,\set{Y}_N$ and  a channel law $P_{Y_1\cdots Y_N|X_1,\ldots,X_N}(y_1,\ldots,y_N|x_1,\ldots,x_N)$, where input $x_j\in\set{X}_j$ and output $y_j\in\set{Y}_j$, for $j\in[1:N]$. At  discrete-time $i\in[1:n]$, Node $j\in[1:N]$  sends  input $x_{j,i}\in \set{X}_j$ and  observes    $y_{j,i}\in\set{Y}_j$, where $n$ denotes the total blocklength used in the transmission.

\begin{figure}[ht]
\centering
\includegraphics[width=0.45\textwidth]{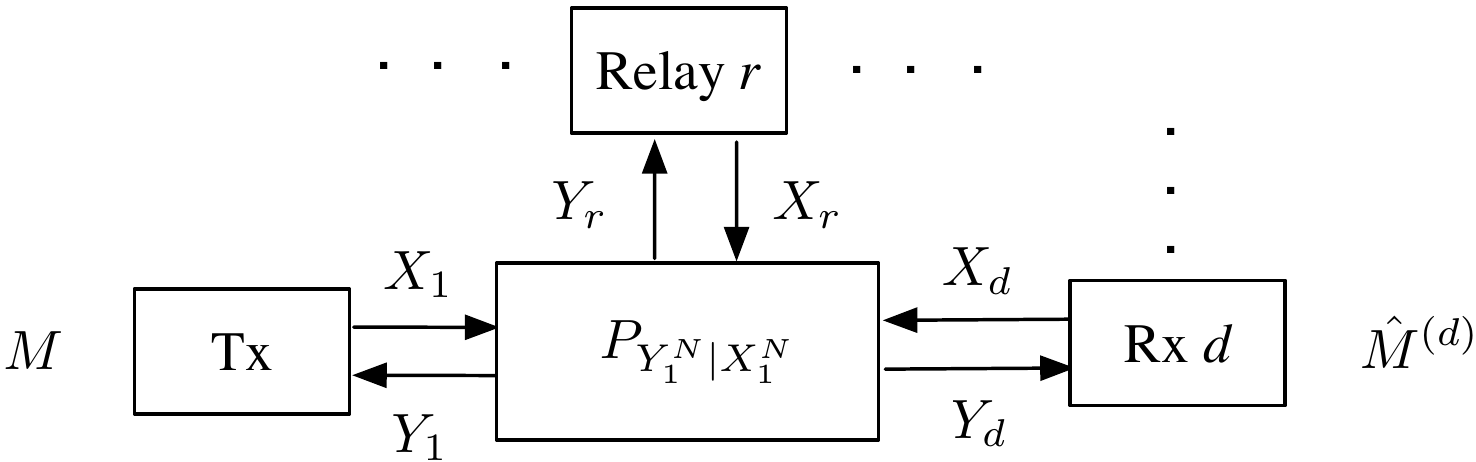}
\caption{$N$-node discrete  memoryless  multicast network } \label{fig:Mrelays}
\end{figure}

The transmitter  communicates a message $m\in[1:2^{nR}]$ to the set of receivers $\set{D}$ with the assistance of   relays  $\set{R}$. A $(2^{nR},n)$  code for this channel has
\begin{itemize}
\item a message set $[1:2^{nR}]$,
\item a source encoder that maps $(M,Y^{i-1}_1)$ to the channel input $X_{1,i}(M,Y^{i-1}_1)$, for each time $i\in[1:n]$,
\item relay and receiver encoders that  maps $Y^{i-1}_k$ to a sequence $X_{k,i}(Y^{i-1}_k)$, for each $k\in[2:N]$ and $i\in[1:n]$,
\item  decoders that   estimates $\hat{M}^{(d)}$ based on $Y^n_d$, for $d\in\mathcal{D}$.
\end{itemize}
Suppose $m$ is uniformly distributed over the message set.  A rate $R$  is called achievable if for every blocklength $n$, there exists  a $(2^{nR},n)$ code such that the average probability of error \[P^{(n)}_e=\text{Pr}{[\hat{M}^{(d)}\neq M,~\text{for some}~d\in\mathcal{D}}]\] tends to 0 as the $n$ tends to infinity. The capacity ${C}_{\text{NoFb}}$  is the supremum of the set of achievable rates $R$ such that  $\lim_{n\to \infty }P_e^{(n)}=0$. 

\subsection{Discrete  memoryless multicast networks with feedback}\label{Model:fb}
\begin{figure}[ht]
\centering
\includegraphics[width=0.45\textwidth]{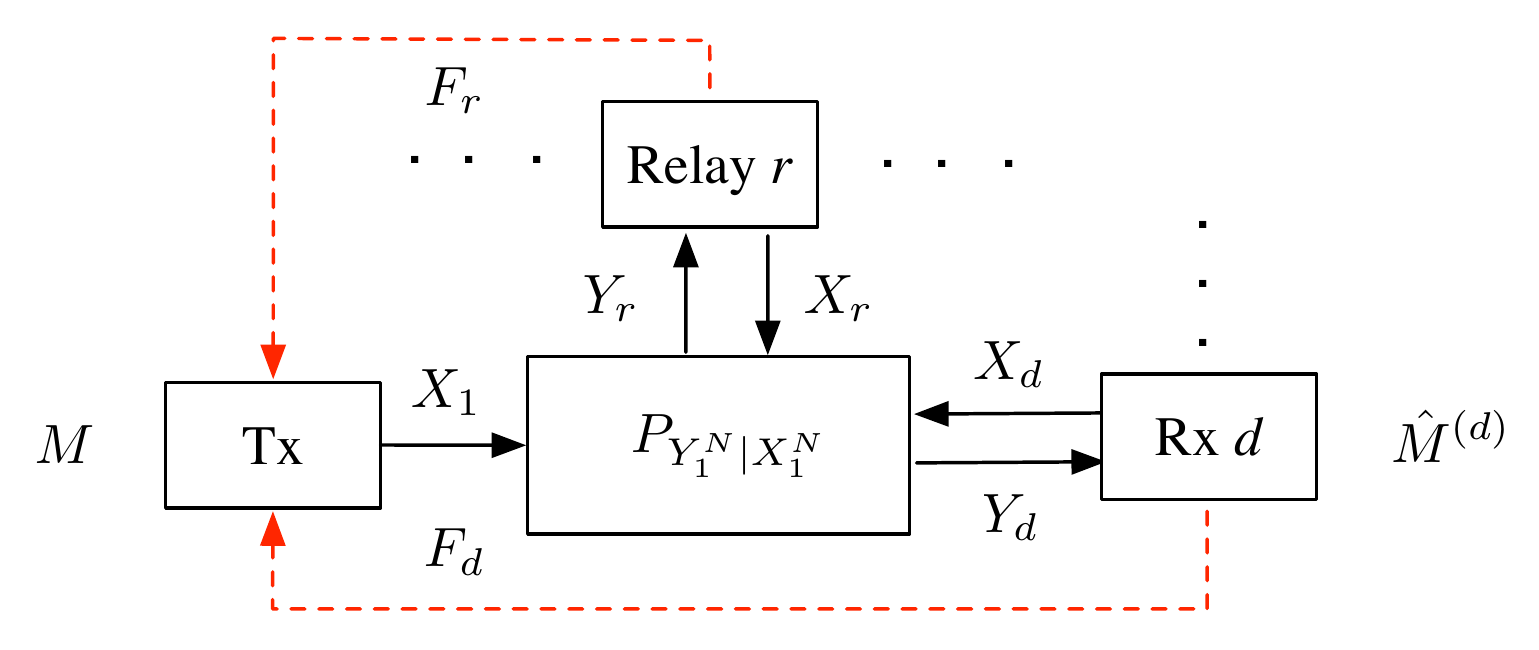}
\caption{$N$-node discrete  memoryless  multicast network with feedback} \label{fig:MrelaysFb}
\end{figure}

Consider an $N$-node DM-MN \emph{with feedback}, see Fig. \ref{fig:MrelaysFb}. 
This model can be modified from the model without feedback in  Section \ref{Model:nofb} if in the later setup $Y_{1,i}$ represents the feedback signals sent by the receivers and relays, i.e.,  $Y_{1,i}=(F_{2,i},\ldots,F_{N,i})$,  where  $F_{k,i}\in \set{F}_{k,i}$ is the feedback signal sent by Node $k\in[2:N]$, and $\set{F}_{k,i}$ denotes feedback finite alphabets. 

More specifically,  after observing $Y_{k,i}$,  each Node $k\in[2:N]$ produces the \emph{channel input} $X_{k,i}$ and \emph{feedback signal} $F_{k,i}$ based on $(Y_{k,1},\ldots,Y_{k,i-1})$,  then it broadcasts $X_{k,i}$ into the network and sends $F_{k,i}$ to the transmitter through a feedback pipe. Let ${C}_{\textnormal{Fb}}$ denote the  capacity  in the presence of  feedback.


 Suppose the feedback pipe  between  the transmitter and Node $k$ is instant, noiseless and \emph{rate-limited} to $R_{\textnormal{Fb},k}$ bits per channel use, then
\begin{IEEEeqnarray}{rCl}\label{consFB0}
|\set{F}_{k,1}|&\times&\cdots\times|\set{F}_{k,n}|\leq 2^{nR_{\textnormal{Fb},k}}, \quad k\in[2:N].
\end{IEEEeqnarray}

When $R_{\text{Fb,k}}\geq \log_2{|\set{Y}_k|}$, the network is equivalent to perfect-feedback setup where after each channel use the receivers and relays send   their channel outputs back to the transmitter, i.e., $ Y_{1}=(Y_{2},\ldots,Y_{N})$.

\section{Previous lower bounds and motivations}
We  recall some previous lower bounds on the capacity of DM-MN and present  interesting observations that inspire our work.

\subsection{NNC and DDF lower bounds}
The NNC  was proposed  by Lim \emph{et al.} \cite{Lim'11}, where each Node $k\in[1:N]$ compresses its received signal $Y_k$ to $\hat{Y}_k$ and  broadcast the compression message into the network. For DM-MN,   the NNC leads to the lower bound which satisfies 
\begin{subequations}\label{eq:NNCorginal}
\begin{IEEEeqnarray}{rCl}\label{eq:NNCorginalRate}
R&\leq & I(X(\set{S});  \hat{Y}(\set{S}^c),Y_d|X(\set{S}^c))\nonumber\\
&&\quad-   I(\hat{Y}(\set{S});Y(\set{S})|X_1^N,Y_d,\hat{Y}(\set{S}^c))
\end{IEEEeqnarray}
for  all $d\in\mathcal{D}$, $\mathcal{S}\subset [1:N]$ with $\mathcal{S}\cap\{1\}\neq \emptyset$, $\mathcal{S}^c\cap\set{D}\neq \emptyset$ and for some pmf 
\begin{IEEEeqnarray}{rCl}\label{eq:pmfNNC}
 \left[\prod^N_{k=1}\! P_{X_k}P_{\hat{Y}_k|X_kY_k}\right]\!P_{Y_1^N|X_1^N}.
\end{IEEEeqnarray}
\end{subequations}

The DDF was proposed by  Lim \emph{et al.} \cite{Lim'14ITW}, where the transmitter uses \emph{backward} encoding to generate auxiliary indices that control the transmission over the entire network. For DM multicast network, the DDF lower bound satisfies 
\begin{IEEEeqnarray}{rCl}
R&\leq& I(X(\set{S});U(\set{S}^c),Y_d|X(\set{S}^c))\nonumber\\
&&\quad - \sum_{k\in\set{S}^c} I(U_k;U(\set{S}_k^c),X^N|X_k,Y_k)
\end{IEEEeqnarray}
for all $d\in\mathcal{D}$, $\mathcal{S}\subset [1:N]$ with $\mathcal{S}\cap\{1\}\neq \emptyset$, $\mathcal{S}^c\cap\set{D}\neq \emptyset$ and for some pmf  $\left[ \prod_{k=2}^N P_{X_k}\right]P_{X_1|X^N_2}P_{Y_1^N|X_1^N} P_{U_2^N|X^N_1}\! $,
where  $\set{S}_k^c=\set{S}^c\cap [1:k-1]$ for $k\in[1:N]$. 


\subsection{Motivation}


Given the NNC lower bound \eqref{eq:NNCorginal}, rewrite the right term of \eqref{eq:NNCorginalRate}: 
\begin{IEEEeqnarray}{rCl}\label{eq:MotiNNC}
&&I(X(\set{S});  \hat{Y}(\set{S}^c),Y_d|X(\set{S}^c))-   I(\hat{Y}(\set{S});Y(\set{S})|X_1^N,Y_d,\hat{Y}(\set{S}^c))\nonumber\\
&&\quad  = I(X_1,X(\set{T});  \hat{Y}(\set{T}^c),Y_d|X(\set{T}^c))\nonumber\\&&\qquad- I(\hat{Y}_1,\hat{Y}(\set{T});Y_1,Y(\set{T})|X_1^N,Y_d,\hat{Y}(\set{T}^c))\nonumber\\
&&\quad  \stackrel{(a)}= I(X_1,\!X(\set{T});  \hat{Y}(\set{T}^c),\!Y_d|X(\set{T}^c))\!-\!I(\hat{Y}_1;Y_1|X_1^N\!,Y_d,\hat{Y}_2^N)\nonumber\\ 
&&\quad\quad- I(\hat{Y}(\set{T});Y(\set{T})|X_1^N,Y_d,\hat{Y}(\set{T}^c))\nonumber\\
&& \quad\stackrel{(b)}\leq  I(X_1,X(\set{T});  \hat{Y}(\set{T}^c),Y_d|X(\set{T}^c))\nonumber\\ 
&&\quad\quad - I(\hat{Y}(\set{T});Y(\set{T})|X_1^N,Y_d,\hat{Y}(\set{T}^c))
\end{IEEEeqnarray}
where $\set{T}=\set{S}\backslash\{1\}$ and  $\mathcal{T}^c$ is the complement of $\mathcal{T}$ in $[2:N]$. The  equality $(a)$ holds because $\big(X([1:N]\backslash \{k\}),\hat{Y}([1:N]\backslash \{k\})\big)-(X_k,Y_k)-\hat{Y}_k$ forms Markov chain in view of pmf \eqref{eq:pmfNNC}, for all $k\in[1:N]$. 

From the inequality $(b)$ in \eqref{eq:MotiNNC}, it's easy to check that the optimal choice of $\hat{Y}_1$ is $\hat{Y}_1=\emptyset$, which makes  NNC lower bound  equivalent to
\begin{IEEEeqnarray}{rCl}\label{eq:NNCrate}
&&R\leq I(X_1,X(\set{T});  \hat{Y}(\set{T}^c),Y_d|X(\set{T}^c))\nonumber\\ 
&&\qquad\quad-~ I(\hat{Y}(\set{T});Y(\set{T})|X_1^N,Y_d,\hat{Y}(\set{T}^c))
\end{IEEEeqnarray}
for  all $d\in\mathcal{D}$, $\mathcal{T}\subset [2:N]$ with $\mathcal{T}^c\cap\set{D}\neq \emptyset$ and for some pmf $\left[\prod^N_{k=1}\! P_{X_k}\right]P_{Y_1^N|X_1^N} \left[\prod^N_{k=2}P_{\hat{Y}_k|X_kY_k}\right].$ Recall that in NNC the auxiliary random variable $\hat{Y}_1$ represents the compression  of $Y_1$   at the transmitter. Setting $\hat{Y}_1=\emptyset$ means that the transmitter  doesn't compress $Y_1$, and simply ignores  it. 

Now consider the DDF scheme.  Due to the backward encoding at the transmitter, the transmitter has to perform encoding offline before receiving $Y_1$, which makes it  impossible to ultilize $Y_1$.

To summarize,   both NNC and DDF fail to  use  the transmitter's observation $Y_1$ (In NNC scheme, although the transmitter can process $Y_1$,   the optimal choice is not to use it.). In fact, ${Y}_1$ is potentially useful to improve the achievable rate and thus should not be simply ignored. In this paper, we propose  new block-Markov coding schemes which improve the NNC and DDF lower bounds by using $Y_1$ and  hybrid relaying strategies. The key idea is as follows: 
\begin{itemize}
\item Each Node $k\in[2:N]$ compresses its channel output $Y_k$ to $\hat{Y}_k$, and decodes (part of) the source message. Then it sends the channel input consisting of the compression message and the decoded  source message.
\item The transmitter obtains or decodes\footnote{When $Y_1$ represents the feedback signals sent by the receivers and relays, the transmitter directly obtains compression messages through feedback pipes; when $Y_1$ represents the channel output instead of feedback signals, the transmitter decodes the compression message based on $Y_1$ with side information $X_1$.} the compression messages  generated by all the receivers and relays, based on $Y_1$, and then forwards them together with the source message.
\item Each Receiver $k$ decodes all compression messages except the one it generated itself,  and then reconstructs $(\hat{Y}_2,\ldots,\hat{Y}_N)$. Finally it uses  $(Y_k,\hat{Y}_2,\ldots,\hat{Y}_N)$ as enhanced outputs to decode the source message.
\end{itemize}

\section{Main Results}
This section  presents our main results. The  proofs are given in Sections \ref{sec:mRelays} and  \ref{Sec:extension}.
\subsection{DM-MN with rate-limited feedback}
\begin{Theorem}\label{Them:unicast1}
 For  DM-MN with feedback from the  receivers and relays to the transmitter, any rate $R>0$ is achievable if it satisfies
\begin{IEEEeqnarray}{rCl}\label{eq:regionThm1}
R&\leq &  I(X_1,X(\set{T});  \hat{Y}(\set{T}^c),Y_d|X(\set{T}^c))\nonumber\\
&&\quad-   I(\hat{Y}(\set{T});Y(\set{T})|X_1^N,Y_d,\hat{Y}(\set{T}^c))
\end{IEEEeqnarray}
for all $d\in\set{D}$ and   $\mathcal{T}\subset [2:N]$ with $\mathcal{T}^c\cap\set{D}\neq \emptyset$,    and for some pmf 
\begin{IEEEeqnarray}{rCl}\label{eq:pmfThm1}
 \left[\prod^N_{k=2} P_{X_k}\right]\left[\prod^N_{k=2}P_{\hat{Y}_k|X_kY_k}\right]P_{X_1|X_2^N}P_{Y_1^N|X_1^N}
\end{IEEEeqnarray}
such that 
\begin{IEEEeqnarray}{rCl} \label{fb:Thm1}
 R_{\textnormal{Fb},k} \geq I(\hat{Y}_k;Y_k|X_k),\quad  \text{for~$k\in[2:N]$},
\end{IEEEeqnarray}
where  $\mathcal{T}^c$ is the complement of $\mathcal{T}$ in $[2:N]$.
\end{Theorem}

\begin{proof}
See Section \ref{sec:scheme1A}.
\end{proof}

\begin{Remark}\label{eq:Remark1} 
Comparing the lower bound in Theorem \ref{Them:unicast1} with the NNC lower bound \eqref{eq:NNCrate}, our rate stirctly includes  NNC  if the feedback rates are sufficient large, i.e., if \eqref{fb:Thm1} holds for all  pmfs \eqref{eq:pmfThm1}. This is  because  in \eqref{eq:pmfThm1} we allow the joint input distribution of form $\prod^N_{k=2}\! P_{X_k}\!P_{X_1|X_2^N}$ instead of $\prod^N_{k=1}\! P_{X_k}$, which attains cooperation between the transmitter and the relays\&receivers.  
\end{Remark}

In the scheme for Theorem \ref{Them:unicast1}, the relays and receivers both perform  compress-forward. In fact, the relays can apply a hybrid strategy that combines compress-forward and partial decode-forward, which leads to a larger achievable rate below.

\begin{Theorem}\label{Them:unicast2}
 For  DM-MN with feedback from the receivers and relays to the transmitter,any rate $R>0$ is achievable if it satisfies
 \begin{subequations}\label{eq:regionThem2}
\begin{IEEEeqnarray}{rCl}
R&\leq & I(X_1,X(\set{T}),U(\set{T});\hat{Y}(\set{T}^c),Y_d|X(\set{T}^c), U(\set{T}^c))\nonumber\\
&&- I(\hat{Y}(\set{T});{Y}(\set{T})|U_2^N,X_1^N,\hat{Y}(\set{T}^c),Y_d)\nonumber\\
&&+\min_{r\in\set{R}} I(U_r;{Y}_r|X_r)\label{eq:Thm2a}\\
R&\leq &  I(X_1,X(\set{T}\cup\set{R}),U(\set{R});\hat{Y}(\set{T}^c\cap\set{D}),Y_d|X(\set{T}^c\cap\set{D}))\nonumber\\
&&- I(\hat{Y}(\set{T}\cup\set{R});{Y}(\set{T}\cup\set{R})|U_2^N, X_1^N,\hat{Y}(\set{T}^c\cap\set{D}),Y_d)\nonumber\\ \label{eq:Thm2b}
\end{IEEEeqnarray}
\end{subequations}
for all $d\in\set{D}$ and   $\mathcal{T}\subset [2:N]$ with $\mathcal{T}^c\cap\set{D}\neq \emptyset$, and for some pmf 
\begin{IEEEeqnarray}{rCl}\label{eq:pmfThem2}
&&  \left[\prod_{r\in\set{R}} P_{X_rU_r}\!P_{\hat{Y}_r|U_rX_rY_r}\! \right]\left[\prod_{d\in\set{D}}\!P_{X_d} P_{\hat{Y}_d|X_dY_d}\right]\nonumber\\
&&\quad\quad\times  P_{X_1|X^N_2U(\set{R})}P_{Y_1^N|X_1^N}
\end{IEEEeqnarray}
 such that
 \begin{subequations}\label{eq:fbTh2}
\begin{IEEEeqnarray}{rCl}
R_{\textnormal{Fb},r} &  \geq&  I(\hat{Y}_r;Y_r|X_r,U_r),\quad \text{ for $r\in\set{R}$}\\
R_{\textnormal{Fb},d} &  \geq& I(\hat{Y}_d;Y_d|X_d),\quad \text{for $d\in\set{D}$}.
\end{IEEEeqnarray}
\end{subequations}
where $U_d=\emptyset$ for all $d\in\set{D}$ and $\mathcal{T}^c$ is the complement of $\mathcal{T}$ in $[2:N]$.
\end{Theorem}

\begin{proof}
See Section \ref{sec:scheme1C}.
\end{proof}



\begin{Remark}\label{eq:Remark2}
By setting $U_r=\emptyset$, for all $r\in\set{R}$, the achievable rate in Theorem \ref{Them:unicast2} specializes to the lower bound in Theorem 1. Note that setting $U_r=\emptyset$ for all $r\in\set{R}$ means  that all relay nodes   perform only compress-forward strategy without partially decode-forwarding the source message.
\end{Remark}

 In the scheme for Theorem \ref{Them:unicast2}, all relays decode the same part of the source message, which may lead to relatively poor rate performance when some relay's observed signal is  bad. One improvement could be made by allowing different relays to decode different parts of the source message according to the strength of their received signals. By doing this way, we obtain the following achievable rate.


\begin{Theorem}\label{Them:unicast2.1}
 For  DM-MN with feedback from the receivers and relays to the transmitter, any rate $R>0$ is achievable if it satisfies \begin{subequations}\label{eq:regionThem2.1}
\begin{IEEEeqnarray}{rCl}
R&\leq & (X_1,X(\set{T}),U(\set{T});\hat{Y}(\set{T}^c),Y_d|V_0,U_0,X(\set{T}^c), U(\set{T}^c))\nonumber\\
&&+ \!\!\!\sum_{r\in\mathcal{T}^c\cap\set{R}} \!\!\! I(U_r;Y_r|U_0,V_0,X_r)+\min_{r\in\set{R}} I(U_0;{Y}_r|V_0,X_r)\nonumber\\
&&- I(\hat{Y}(\set{T});{Y}(\set{T})|V_0,U_0,U_2^N,X_1^N,\hat{Y}(\set{T}^c),Y_d)\label{eq:Thm2.1a}\\
R&\leq & I(V_0,U_0,X_1,X(\set{T}\cup\set{R}),U(\set{R});\hat{Y}(\set{T}^c\cap\set{D}),Y_d|\nonumber\\
&&\hspace{4cm}X(\set{T}^c\cap\set{D}))\nonumber\\
&&- I(\hat{Y}(\set{T}\cup\set{R});{Y}(\set{T}\cup\set{R})|V_0,U_0,U_2^N,\nonumber\\
&&\hspace{4cm} X_1^N,\hat{Y}(\set{T}^c\cap\set{D}),Y_d) \label{eq:Thm2.1b}
\end{IEEEeqnarray}
\end{subequations}
for all   $d\in\set{D}$ and   $\mathcal{T}\subset [2:N]$ with $\mathcal{T}^c\cap\set{D}\neq \emptyset$, and for some pmf 
\begin{IEEEeqnarray}{rCl}\label{eq:pmfThem2.1}
&&P_{V_0}P_{U_0|V_0}\left[\prod_{r\in\set{R}} P_{X_r|V_0}P_{U_r|V_0U_0X_r}P_{\hat{Y}_r|V_0U_0U_rX_rY_r} \right] \nonumber\\
&&\quad\times \left[\prod_{d\in\set{D}} P_{X_d}P_{\hat{Y}_d|X_dY_d}\right]P_{X_1|V_0U_0X^N_2U(\set{R})}P_{Y_1^N|X_1^N}
\end{IEEEeqnarray}
 such that
 \begin{subequations}\label{eq:fbTh2.1}
\begin{IEEEeqnarray}{rCl}
R_{\textnormal{Fb},r} &  \geq&  I(\hat{Y}_r;Y_r|V_0,U_0,X_r,U_r),\quad \text{ for $r\in\set{R}$}\\
R_{\textnormal{Fb},d} &  \geq& I(\hat{Y}_d;Y_d|X_d),\quad \text{for $d\in\set{D}$}
\end{IEEEeqnarray}
\end{subequations}
where $U_d=\emptyset$ for all $d\in\set{D}$ and $\mathcal{T}^c$ is the complement of $\mathcal{T}$ in $[2:N]$.
\end{Theorem}
\begin{proof}
See Section \ref{sec:scheme1C}.
\end{proof}

\begin{Remark}\label{RemarkTheorem3}
By letting $V_0=U_0=\emptyset$,  we find that  the constraint \eqref{eq:Thm2.1a} is less stringent than \eqref{eq:Thm2a};    \eqref{eq:Thm2.1b}, \eqref{eq:pmfThem2.1} and \eqref{eq:fbTh2.1} reduce to \eqref{eq:Thm2b}, \eqref{eq:pmfThem2} and \eqref{eq:fbTh2}, respectively. Thus the rate  in Theorem \ref{Them:unicast2.1} includes the rate in Theorem \ref{Them:unicast2}.
\end{Remark}

\subsection{DM-MN {without}  feedback}

\begin{Remark}\label{Mark:Fb2NoFb}
The lower bounds present above can be directly extended to DM-MN \emph{without} feedback by letting feedback rate $R_{\textnormal{Fb},k}=0$, and $\hat{Y}_k=\emptyset$ for all $k\in[2:N]$.\end{Remark} In this subsection, we  propose a new lower bound for DM-MN {without} feedback that makes use of  the channel outputs (not feedback signals) observed at the transmitter. The key idea is that  the transmitter, instead of obtaining compression messages  through feedback,   \emph{decodes} them  based on its observed channel outputs.  The  new achievable rate is shown   below.


\begin{Theorem}\label{Them:unicast3}
For   DM-MN without feedback, any rate $R>0$ is achievable if it satisfies
\begin{subequations}\label{eq:regionThem3}
\begin{IEEEeqnarray}{rCl}
R&\leq &  I(X_1,X(\set{T}),U(\set{T}),V(\set{T});\hat{Y}(\set{T}^c),Y_d|\nonumber\\
&&\hspace{4cm} U(\set{T}^c),V(\set{T}^c),X(\set{T}^c))\nonumber\\
&& -I(\hat{Y}(\set{T});{Y}(\set{T})|U_2^N,V_2^N,X_1^N,\hat{Y}(\set{T}^c),Y_d)\nonumber\\
&&+ \min_{r\in\set{R}} I(U_r;\hat{Y}_r|X_r,V_r)\\
R
&<&  I(X_1,V(\set{T}\cup\set{R}),U(\set{R}),X(\set{T}\cup\set{R});\hat{Y}(\set{T}^c\cap\set{D}),Y_d|\nonumber\\
&&\hspace{4cm}V(\set{T}^c),X(\set{T}^c\cap\set{D}))\nonumber\\
&&-I(\hat{Y}(\set{T}\cup\set{R});Y(\set{T}\cup\set{R})|V_2^N\!,X_1^N\!,U_2^N\!,\hat{Y}(\set{T}^c\cap\set{D}),Y_d)\nonumber\\
\end{IEEEeqnarray}
\end{subequations}
for all  $d\in\mathcal{D}$ and  $\mathcal{T}\subset [2:N]$ with $\mathcal{T}^c\cap\set{D}\neq \emptyset$, and for some pmf 
\begin{IEEEeqnarray}{rCl}\label{eq:pmfThem3}
&& \left[\prod_{k=2}^N P_{V_k}P_{X_k|V_k}P_{U_k|V_k} \right]\left[\prod_{r\in\set{R}}\!  P_{\hat{Y}_r|U_rV_rX_rY_r}\right]\nonumber\\
&&\quad~\times \left[\prod_{d\in\set{D}} \!P_{\hat{Y}_d|V_dX_dY_d}\right]P_{X_1|V^N_2U(\set{R})}P_{Y_1^N|X_1^N}
\end{IEEEeqnarray}
such that
\begin{IEEEeqnarray}{rCl}\label{eq:fbThm3}
\sum_{r\in\set{T} \cap\set{R} } &&I(\hat{Y}_r;Y_r|U_r,V_r,X_r) + \sum_{d\in\set{T}\cap\set{D} }  I(\hat{Y}_d;Y_d|V_d,X_d)\nonumber\\
 &&\leq I(X(\set{T});Y_1|U_2^N,V_2^N,X(\set{T}^c),X_1)
 \end{IEEEeqnarray}
where $U_d=\emptyset$, for all $d\in\set{D}$, and $\mathcal{T}^c$ is the complement of $\mathcal{T}$ in $[2:N]$.

\begin{proof}
See Section \ref{Sec:extension}.
\end{proof}

\end{Theorem}
\begin{Remark}
The scheme for Theorem \ref{Them:unicast3} requires the transmitter  to decode the  compression messages generated by all receivers and relays, which may limit the  performance if there are  weak links from the receivers or relays to the transmitter. One proper way to improve the scheme is to allow the transmitter   to adaptively decode a set of (not all)   nodes' compression messages. 

Suppose the transmitter  decodes only the compression messages generated by the set of nodes $\set{A}\subseteq [2:N]$. Then by a scheme  similar to that  for Theorem \ref{Them:unicast3}, we obtain a new lower bound having same rate espression as \eqref{eq:regionThem3}, but is maximized over all set ${\set{A}\subseteq [2:N]} $ 
for all  $d\in\mathcal{D}$, $\mathcal{T}\subset [2:N]$ with $\mathcal{T}^c\cap\set{D}\neq \emptyset$,  and for some pmf 
\begin{IEEEeqnarray}{rCl}\label{eq:pmfThem4}
&&\left[\prod_{k=2}^N P_{V_k}P_{X_k|V_k}P_{U_k|V_k}\right]\left[\prod_{r\in\set{R}}\! P_{\hat{Y}_r|U_rV_rX_rY_r}\right]\nonumber\\
&&\quad\times \left[\prod_{d\in\set{D}} \!P_{\hat{Y}_d|V_dX_dY_d}\right]P_{Y_1^N|X_1^N}P_{X_1|V(\set{A})U(\set{A})}
\end{IEEEeqnarray}
such that
\begin{IEEEeqnarray}{rCl}\label{eq:fbThm3}
\sum_{r\in\set{T}_\set{A} \cap\set{R} } &&I(\hat{Y}_r;Y_r|U_r,V_r,X_r) + \sum_{d\in\set{T}_\set{A}\cap\set{D} }  I(\hat{Y}_d;Y_d|V_d,X_d)\nonumber\\
 &&\leq I(X(\set{T}_\set{A});Y_1|U(\set{A}),V(\set{A}),X(\set{T}_\set{A}^c),X_1)
 \end{IEEEeqnarray}
 where $U_d=\emptyset$ for all $d\in\set{D}$, $\mathcal{T}^c$ is the complement of $\mathcal{T}$ in $[2:N]$, $\set{T}_\set{A}=\set{T}\cap\set{A}$, and $\set{T}^c_\set{A}$ is the complement of $\set{T}_\set{A}$ in $\set{A}$. 

This lower bound reduces to the lower bound in Theorem \ref{Them:unicast3} when  $\set{A}=[2:N]$, and to the NNC lower bound when $\set{A}=\emptyset$ and $V_k=U_k=\emptyset$, for all $k\in[2:N]$.
 
\end{Remark}

\section{Examples}\label{sec:Eg}
\subsection{The  relay channel with relay-transmitter feedback}

 Consider the  relay channel $P_{Y_2Y_3|X_1X_2}(y_2,y_3|x_1,x_2)$ with perfect feedback from the relay to the transmitter, see Fig. \ref{fig:relayFb}.

\begin{figure}[h!]
\centering
\includegraphics[width=0.45\textwidth]{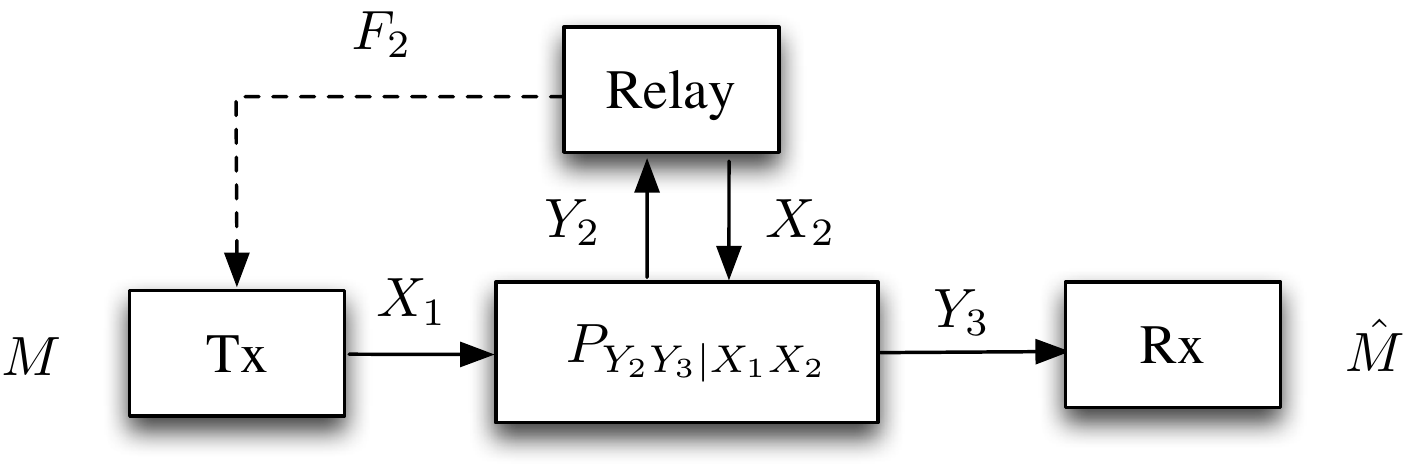}
\caption{Relay channel with relay-transmitter feedback} \label{fig:relayFb}
\end{figure}

 Let $\hat{Y}_3=\emptyset$, then  lower bound in Theorem \ref{Them:unicast1} specializes to 
\begin{subequations}\label{eq:regionCF}
\begin{IEEEeqnarray}{rCl} 
R&\leq& I(X_1; \hat{Y}_2,Y_3|X_2) \\
R&\leq& I(X_1,X_2;Y_3) -I(\hat{Y}_2;Y_2|X_1,X_2,Y_3) 
\end{IEEEeqnarray}
\end{subequations}
for some pmf $P_{X_1X_2}P_{\hat{Y}_2|X_2Y_2}$. 
 
 Let $U_3=\hat{Y}_3=\emptyset$, then  Theorem \ref{Them:unicast2} specializes to 
\begin{subequations}\label{eq:regionCFDF}
\begin{IEEEeqnarray}{rCl} 
R&\leq& I(X_1; \hat{Y}_2,Y_3|U_2,X_2)+I(U_2;Y_2|X_2) \\
R&\leq& I(X_1,X_2;Y_3) -I(\hat{Y}_2;Y_2|U_2,X_1,X_2,Y_3) 
\end{IEEEeqnarray}
\end{subequations}
for some pmf $P_{X_1X_2U_2}P_{\hat{Y}_2|X_2U_2Y_2}$.

In \cite{Gabbai'06} Gabbai and Bross studied this channel and proposed coding schemes based on  restricted decoding and deterministic partitioning. The rates \eqref{eq:regionCF} and \eqref{eq:regionCFDF} recover  Gabbai and Bross's  rates of Theorems 2 and  3 in \cite{Gabbai'06}, respectively.

By using {NNC} \cite{Lim'11},  the  rate $R$ satisfying
\begin{subequations}\label{eq:NNC1R}
\begin{IEEEeqnarray}{rCl}
R\leq &&I(X_1;\hat{Y}_2,Y_3|X_2)\\
R\leq && I(X_1,X_2;Y_3)\!-\!I(\hat{Y}_2;Y_2|X_1,X_2,Y_3)
\end{IEEEeqnarray}
\end{subequations}
is achievable for any pmf $P_{X_1}P_{X_2}P_{\hat{Y}_2|X_2Y_2}$, which coincides with the compress-forward lower bound  \cite[Theorem 6]{Cover'79}.

By using  {DDF} \cite{Lim'14ITW,Lim'14ISIT},  the  rate $R$ satisfying\begin{subequations}\label{eq:regionDistDF}
\begin{IEEEeqnarray}{rCl}
R&\leq& I(X_1,X_2;Y_3)\\
R&\leq& I(U_2;Y_2|X_2)+I(X_1;Y_3|X_2,U_2) 
\end{IEEEeqnarray}
\end{subequations}
is achievable for any pmf  $P_{X_1X_2U_2}$, which coincides with the partial decode-forward lower bound   \cite[Theorem 7]{Cover'79}.

The lower bound \eqref{eq:regionCF} includes \eqref{eq:NNC1R} because it allows a joint input distribution $P_{X_1X_2}$ rather than $P_{X_1} P_{X_2}$. The lower bound \eqref{eq:regionCFDF} includes   \eqref{eq:NNC1R} and \eqref{eq:regionDistDF}, which can be seen by letting $U_2=\emptyset$ and $\hat{Y}_2=\emptyset$, respectively. In \cite{Gabbai'06} Gabbai and Bross showed that for the  Gaussian  and $Z$ relay channels, the lower bound \eqref{eq:regionCFDF}  strictly improves on the known lower bounds on the achievable rate in the absence of feedback, including the compress-forward lower bound  in \eqref{eq:NNC1R}, and the partial decode-forward lower bound  in \eqref{eq:regionDistDF}. In view of this fact, we  have the following corollary:

\begin{Corollary}
For the  DM single-relay channel with relay-transmitter feedback, our coding scheme  recovers Gabbai and Bross's results, and  can strictly improve  on NNC \cite{Lim'11},      DDF  \cite{Lim'14ITW} and all known lower bounds on the achievable rate in the absence of feedback. 
\end{Corollary}
\subsection{Enhanced Gaussian relay channel}

\begin{figure}[h!]
\centering
\includegraphics[width=0.45\textwidth]{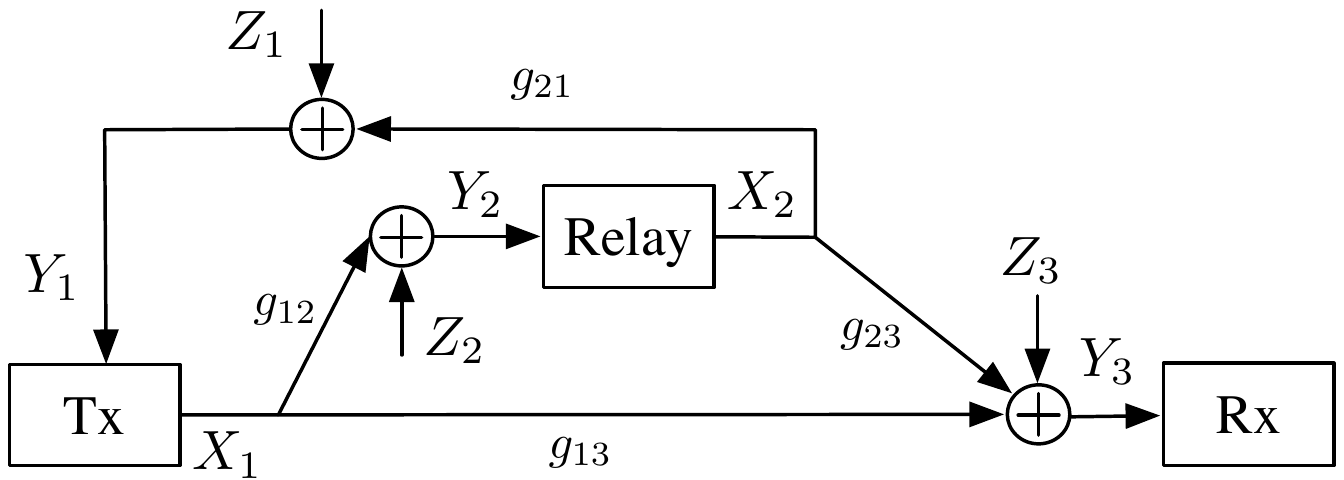}
\caption{The enhanced Gaussian relay channel} \label{fig:EGaussianRelay}
\end{figure}
Consider an enhanced Gaussian relay channel where the transmitter can access the output $Y_1$, see Fig. \ref{fig:EGaussianRelay}. The channel outputs are: 
\begin{IEEEeqnarray*}{rCl}
Y_1&=&g_{21}X_2+Z_1,\\
Y_2&=&g_{12}X_1+Z_2,\\
Y_3&=&g_{13}X_1+g_{23}X_2+Z_3
\end{IEEEeqnarray*}
where $g_{21}$, $g_{23}$, $g_{12}$ and $g_{13}$ are channel gains and $Z_1\sim\set{N}(0,1)$, $Z_2\sim\set{N}(0,1)$ and $Z_3\sim\set{N}(0,1)$ are independent Gaussian noise variables. The  input power constraints  are $\mathbb{E}|X^2_1|\leq P_1$ and $\mathbb{E}|X^2_2|\leq P_2$.  Let $s_{12}=g^2_{12}P_1$, $s_{13}=g^2_{13}P_1$, $s_{23}=g^2_{23}P_2$ and $s_{21}=g^2_{21}P_2$.

We compare the lower bound in Theorem \ref{Them:unicast3}  with the cut-set outer bound and the previous known lower bounds, such as amplify-forward, NNC, DDF  and Cover\--El Gama's general lower bound \cite[Theorem 7]{Cover'79}. \\
\emph{Achievable rate in Theorem \ref{Them:unicast3}:}
Let $U_3=V_3=\hat{Y}_3=\emptyset$, then Theorem \ref{Them:unicast3} reduces to 
\begin{subequations}\label{eq:Erelaypro}
\begin{IEEEeqnarray}{rCl}
R &\leq &I(X_1;\hat{Y}_2,Y_3|U_2,V_2,X_2)+ I(U_2;Y_2|V_2,X_2) \\
R &\leq &I(X_1,X_2;Y_3)-I(\hat{Y}_2;Y_2|U_2,V_2,X_1,X_2,Y_3)\quad
\end{IEEEeqnarray}
for some pmf $P_{V_2}P_{X_2|V_2}P_{U_2|V_2}P_{X_1|V_2U_2}P_{\hat{Y}_2|X_2V_2U_2Y_2}$ such that \begin{IEEEeqnarray}{rCl}I(\hat{Y}_2;Y_2|U_2,V_2,X_2)\leq I(X_2;Y_1|U_2,X_1,V_2).
\end{IEEEeqnarray}
\end{subequations}
To compute  \eqref{eq:Erelaypro},  we choose the same distributions  as in \cite{Chong'05}:
\begin{IEEEeqnarray}{rCl}
&&U_2=aV_2+W_0, X_2=cV_2+W_2\nonumber\\
&&X_1=bU_2+W_1, \hat{Y}_2=Y_2+Z'
\end{IEEEeqnarray}
where $V_2\sim\set{N}(0, P_1), W_0\sim\set{N}(0,\frac{\bar{\alpha}\beta P_1}{b^2}),W_1\sim\set{N}(0,\alpha P_1), W_2\sim\set{N}(0,\gamma P_2)$ and $Z'\sim\set{N}(0,N')$  are independent, for  $\alpha,\beta, \gamma \in[0,1]$. For this choice, we have,
\begin{IEEEeqnarray}{rCl}
I(X_1;\hat{Y}_2,Y_3|X_2,V_2,U_2)&=&\mathcal{C}\Big(\alpha s_{13}+\frac{\alpha s_{12}}{1+N'}\Big) \nonumber\\	
 I(U_2;Y_2|V_2,X_2)&=& \mathcal{C}\Big({\frac{s_{12}\beta\bar{\alpha}}{\alpha s_{12}+1}}\Big)\nonumber\\
 I(X_1,X_2;Y_3) &=& \mathcal{C}\Big({2\sqrt{\bar{\alpha}\bar{\beta}\bar{\gamma}s_{13}s_{23} }\!+\!s_{13}\!+\!s_{23} }\Big)\nonumber\\
 I(\hat{Y}_2;Y_2|U_2,\!V_2,\!X_1,\!X_2,\!Y_3)&=& \mathcal{C}\Big(\frac{1}{N'}\Big),~
\end{IEEEeqnarray}
and 
\begin{IEEEeqnarray}{rCl}
I(\hat{Y}_2;Y_2|U_2,V_2,X_2)&=& \mathcal{C}\Big(\frac{1+\alpha s_{12}}{N'} \Big)\nonumber\\
I(X_2;Y_1|U_2,X_1,V_2) &=& \mathcal{C} ({\gamma s_{21}}).
\end{IEEEeqnarray}
Thus we obtain the lower bound
\begin{IEEEeqnarray}{rCl} \label{eq:pro}
R\leq && \min\Big \{\mathcal{C}\Big(\alpha s_{13}+\frac{\alpha s_{12}}{1+N'}\Big) +\mathcal{C}\Big({\frac{s_{12}\beta\bar{\alpha}}{\alpha s_{12}+1}}\Big), \nonumber\\
&&\quad \mathcal{C}\Big({2\sqrt{\bar{\alpha}\bar{\beta}\bar{\gamma}s_{13}s_{23} }+s_{13}+s_{23} }\Big)- \mathcal{C}\Big(\frac{1}{N'}\Big)\Big \}
\end{IEEEeqnarray}
subject to the constraint
\begin{IEEEeqnarray}{rCl}\label{eq:constPro3}
N' \geq \frac{1+\alpha s_{12}}{\gamma s_{21}}.
\end{IEEEeqnarray}
\emph{Amplify-forward:}  For the general Gaussian relay channel with linear relaying functions, finding the channel capacity is a non-convex optimization problem for blocklength $n\geq 2$, which  is almost  intractable.  The paper \cite{Zahedi'thesis} proposed an achievable rate:
\begin{IEEEeqnarray*}{rCl}
R \leq  \max_{0<\alpha\leq 1}\frac{1}{2}\mathcal{C}\bigg( 2\alpha P \Big(1+\frac{\big(\sqrt{(1-\alpha)/\alpha}+g_{12}g_{23}d\big)^2}{1+g^2_{23}d^2}\Big) \bigg)
\end{IEEEeqnarray*}
where $d=\sqrt{2P_2/(2\alpha s^2_{13}+1)}$.\\
\emph{NNC:}
When using NNC \cite{Lim'11}, the achievable rate is:
\begin{IEEEeqnarray*}{rCl}\label{eq:NNC1REnhance}
R&\leq &I(X_1;\hat{Y}_2,Y_3|X_2)-I(\hat{Y}_1;Y_1|X_1,X_2,\hat{Y}_2,Y_3),\nonumber\\
R&\leq &  I(X_1,X_2;Y_3)\!-\!I(\hat{Y}_2;Y_2|X_2,\!Y_3)\!-\!I(\hat{Y}_1;Y_1|X_1,\!X_2,Y_3)
\end{IEEEeqnarray*}
for some pmf $P_{X_1}P_{X_2}P_{\hat{Y}_2|X_2Y_2}P_{\hat{Y}_1|X_1Y_1}$.  It's easy to check that  the optimal choice of $Y_1$ is $\hat{Y}_1=\emptyset$, which leads to the compress-forward lower bound   \eqref{eq:NNC1R}. The optimal distribution of $\hat{Y}_2$ is generally unknown. Choose $\hat{Y}_2=Y_2+Z'$ where $Z'\sim\set{N}(0,\sigma^2)$ and optimise over $\sigma^2$. We obtain the achievable rate
\begin{IEEEeqnarray}{rCl}
R\leq \mathcal{C}\Big(s_{13}\!+\!\frac{s_{12}s_{23}}{s_{13}\!+\!s_{12}\!+\!s_{23}\!+\!1}\Big).
\end{IEEEeqnarray} 
\emph{DDF:}
When using {DDF} \cite{Lim'14ITW}, the achievable rate is same as the partial  decode-forward lower bound \eqref{eq:regionDistDF}. For the Gaussian relay channels, partial decode-forward coding doesn't improve the decode-forward  lower bound \cite{Zahedi'thesis},   thus we obtain the achievable rate
\begin{IEEEeqnarray}{rCl}
R\leq  \min\big \{\mathcal{C}(s_{13}\!+\!s_{23}\!+\!2\rho\sqrt{s_{13}s_{23}}),\mathcal{C}\big(s_{12}(1\!-\!\rho^2)\big) \big\}~
\end{IEEEeqnarray}
for $0\leq \rho\leq 1$.\\
\emph{ Cover--El Gamal's general lower bound \cite[Theorem 7]{Cover'79}:}   
In \cite{Cover'79} Cover and El Gamal proposed a general lower bound for the  relay channel by  combining compress-forward and decode-forward, which can be written as:
\begin{subequations}
\begin{IEEEeqnarray}{rCl}\label{eq:ErelayCover}
R &\leq& I(X_1;\hat{Y}_2,Y_3|X_2,U_2)+ I(U_2;Y_2|V_2,X_2) \\
R&\leq&I(X_1,X_2;Y_3)-I(\hat{Y}_2;Y_2|U_2,X_1,X_2,Y_3)
\end{IEEEeqnarray}
for some pmf $P_{V_2}P_{X_2|V_2}P_{U_2|V_2}P_{X_1|U_2}P_{\hat{Y}_2|X_2U_2Y_2}$ such that \begin{IEEEeqnarray}{rCl}I(\hat{Y}_2;Y_2|U_2,X_1,X_2,Y_3)\leq I(X_2;Y_3|V_2).\end{IEEEeqnarray}
\end{subequations}

Choosing the same distributions  as in \cite{Chong'05},  we obtain the lower bound with same  expression as \eqref{eq:pro} but subject to the  constraint 
\begin{IEEEeqnarray}{rCl}\label{eq:cover}
N'\geq (\alpha ( s_{13}+s_{23})+1)\frac{(\beta-\alpha\beta+\alpha)s_{13}+1 }{\gamma s_{23}(\alpha s_{13}+1)}.
\end{IEEEeqnarray}

Comparing \eqref{eq:constPro3} with \eqref{eq:cover}, if
\begin{IEEEeqnarray}{rCl} \label{condition:enh}
\frac{1+\alpha s_{12}}{ s_{21} (\alpha ( s_{13}+s_{23})+1)} <\frac{(\beta-\alpha\beta+\alpha)s_{13}+1 }{ s_{23}(\alpha s_{13}+1)}
\end{IEEEeqnarray}
 for all $\alpha,\beta\in[0,1]$ (e.g.  $s_{21}>s_{23}, s_{12}<s_{13}$), our coding scheme always improves Cover--El Gama's general lower bound \cite[Theorem 7]{Cover'79}. This general lower bound includes both the partial decode-forward  and  compress-forward lower bounds \cite{Cover'79}, thus we have the following corollary:
\begin{Corollary}
For the enhanced Gaussian relay channel which  satisfies \eqref{condition:enh}, our coding scheme  improves the known inner bounds, including  the NNC and DDF lower bounds  and Cover--El Gama's general lower bound \cite[Theorem 7]{Cover'79}. 
\end{Corollary}


Based on (\ref{eq:pro}--\ref{eq:cover}), the achievable rates  for $ g_{12}=g_{13}=g_{21}=1$, $g_{23}=0.7$, and $P_1=P_2=P$ are shown in Fig. \ref{fig:RateERelay}. 

\begin{figure}[h!]
\centering
\includegraphics[width=0.5\textwidth]{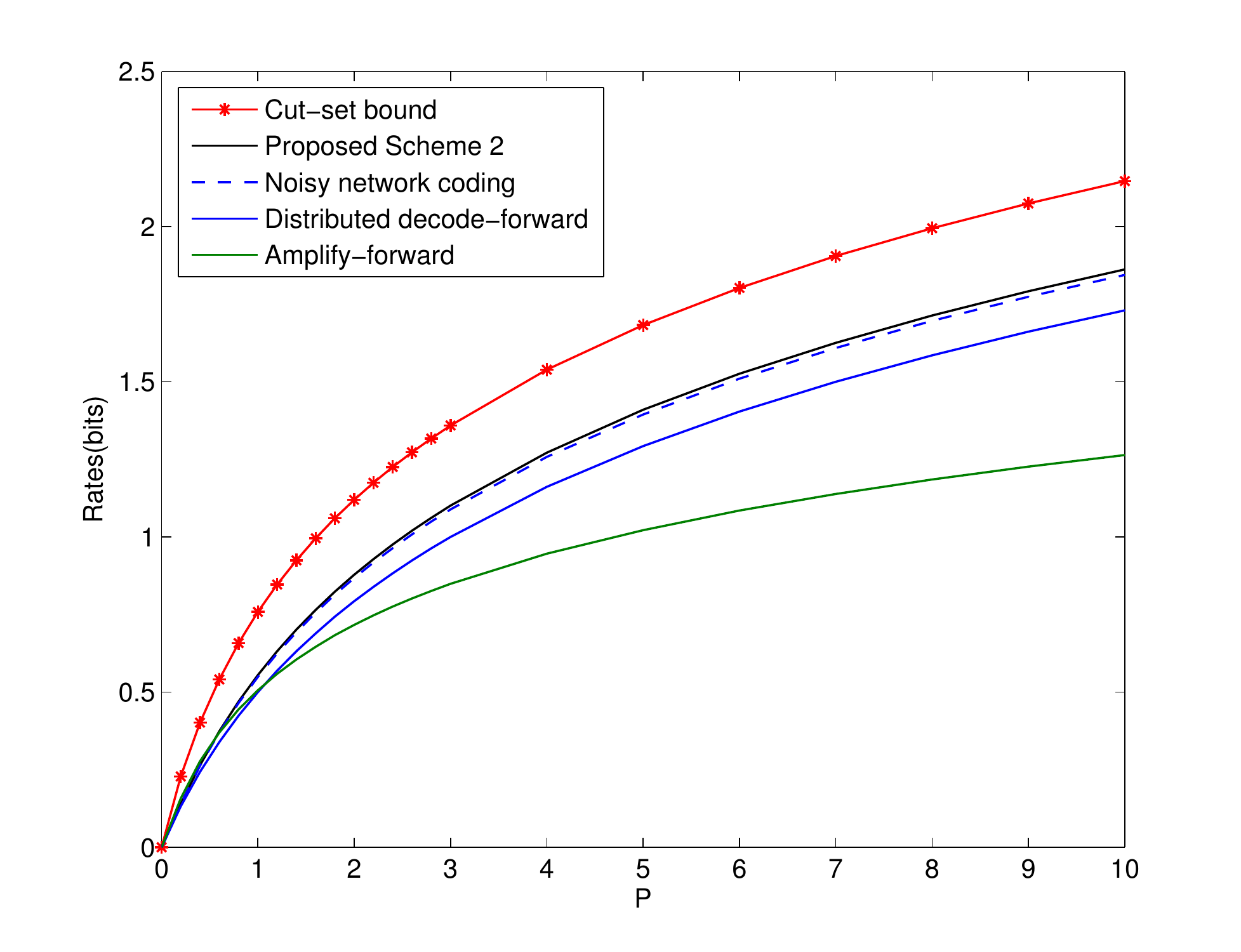}
\caption{Achievable rates for the enhanced Gaussian relay channel with $ g_{12}=g_{13}=g_{21}=1$, $g_{23}=0.7$, and $P_1=P_2=P$} \label{fig:RateERelay}
\end{figure}

Table \ref{tab:compareRates} compares achievable rates for this enhanced Gaussian relay channel for $g_{12}=1/d, g_{13}=1$, $g_{23}=g_{21}=1/|1-d|$, and with $P_1=5, P_2=1$. Here $R_\text{NNC}$, $R_\text{DDF}$, $R_\text{CE}$, $R_\text{Pro1}$ and $R_\text{Pro2}$ denote  rates achieved by NNC, DDF,  rates from \cite[Theorem 7]{Cover'79} and rates from our proposed Theorem \ref{Them:unicast2} and \ref{Them:unicast3}, respectively. The  feedback scheme ($R_\text{Pro1}$) obtains the best performance, and our non-feedback scheme for Theorem  \ref{Them:unicast3} ($R_\text{Pro2}$) strictly improves  the known lower bounds in the absence of feedback.

\begin{table}
\begin{center}
\caption{Achievable rates for the enhanced Gaussian relay  with and without  feedback}
\begin{tabular}{lcccccc}
\toprule
$d$ & $R_\text{NNC}$ & $R_\text{DDF}$ & $R_\text{CE}$ & $R_\text{Pro1}$& $R_\text{Pro2}$\\
\midrule
0.73& 1.6908 & 1.6881& 1.6927 & 1.7069 & 1.6996\\
0.74& 1.6971 & 1.6703& 1.6971 & 1.7111 & 1.7032\\
0.75& 1.7033 & 1.6529& 1.7033 & 1.7153 & 1.7077 \\
0.76& 1.7094 & 1.6358& 1.7094 & 1.7195 & 1.7129\\
\bottomrule
 \label{tab:compareRates}
\end{tabular}
\end{center}
\end{table}

\subsection{4-Node Discrete Memoryless Diamond Network}
Consider the 4-node DM diamond network $P_{Y_2Y_3|X_1}(y_2,y_3|x_1)P_{Y_4|X_2X_3}(y_4|x_2,x_3)$ \cite{Schein'00}.  From Theorem \ref{Them:unicast2.1},   we have 

\begin{Corollary}
For 4-node DM diamond network with relay-transmitter feedback, any rate $R>0$ is achievable if it satisfies
\begin{subequations}\label{eq:DiamonRateFb}
\begin{IEEEeqnarray}{rCl}
R &\leq& I(X_1,X_2,X_3; Y_4)\nonumber\\
&&-I(\hat{Y}_2,\hat{Y}_3;Y_2,Y_3|V_0,U_0,U_2,U_3,X_1,X_2,X_3,Y_4)\\
R &\leq& I(X_1,X_2,U_2;\hat{Y}_3,Y_4|V_0,U_0,X_3,U_3)\nonumber\\
&& +I(U_3;Y_3|V_0,U_0,X_3)+\min_{r\in\{2,3\} }I(U_0;Y_r|V_0,X_r)\nonumber\\
&&-I(\hat{Y}_2;Y_2|V_0,U_0,U_2,U_3,X_1,X_2,X_3,\hat{Y}_3,Y_4)\\
R &\leq& I(X_1,X_2,U_3;\hat{Y}_2,Y_4|V_0,U_0,X_2,U_2)\nonumber\\
&& +I(U_2;Y_2|V_0,U_0,X_2)+\min_{r\in\{2,3\} }I(U_0;Y_r|V_0,X_r)\nonumber\\
&&-I(\hat{Y}_3;Y_3|V_0,U_0,U_2,U_3,X_1,X_2,X_3,\hat{Y}_2,Y_4)\\
R &\leq& I(X_1,X_2,U_2,U_3;Y_4|V_0,U_0)+  \min_{r\in\{2,3\} }I(U_0;Y_r|V_0,X_r)\nonumber\\
&&-I(\hat{Y}_2,\hat{Y}_3;Y_2,Y_3|V_0,U_0,U_2,U_3,X_1,X_2,X_3,Y_4)\nonumber\\
R &\leq& I(X_1; \hat{Y}_2,\hat{Y}_3,Y_4|V_0,U_0,X_2,X_3,U_2,U_3)\nonumber\\
&& +I(U_2;Y_2|U_0,V_0,X_2)+I(U_3;Y_3|U_0,V_0,X_3)\nonumber\\ &&+\min_{r\in\{2,3\} }I(U_0;Y_r|V_0,X_r)
\end{IEEEeqnarray}
\end{subequations}
for  some pmf  $  P_{U_0V_0}P_{X_1|V_0U_0X_2X_3U_2U_3}\Big[ \prod_{r\in\{2,3\}} P_{X_r|V_0}$\\$P_{U_r|V_0U_0X_r}P_{\hat{Y}_r|V_0U_0U_rX_rY_r} \Big] $ such that satisfies
$R_{\textnormal{Fb},r}   \geq  I(\hat{Y}_r;Y_r|V_0,U_0,X_r,U_r)$, for $r\in\{2,3\}$.
\end{Corollary}

From Remark  \ref{eq:Remark1}, \ref{eq:Remark2}  and \ref{RemarkTheorem3},  we  know that the rate \eqref{eq:DiamonRateFb} strictly includes the NNC lower bound if feedback rate is sufficiently large. Now consider the DDF  lower bound for this setup, which is the rate $R>0$ satisfying  
\begin{subequations}\label{eq:DiamonRateDDF}
\begin{IEEEeqnarray}{rCl}
R &\leq& I(X_1,X_2,X_3; Y_4)\\
R &\leq& I(X_1,X_2;{Y}_4,|X_2,U_2)+I(U_2;Y_2|X_2)\\
R &\leq& I(X_1,X_3;{Y}_4,|X_3,U_3)+I(U_3;Y_3|X_3)\\
R &\leq& I(X_1; U_2,U_3,Y_4|X_2,X_3)-I(U_2;X_1,X_3|X_2,Y_2)\nonumber\\
&& -I(U_3;X_1,X_2,U_2|X_3,Y_3) \label{eq:Diff2}
\end{IEEEeqnarray}
for some pmf $P_{X_2}P_{X_3}P_{X_1|X_2X_2}P_{U_2U_3|X_1X_2X_3}$
\end{subequations}
It's not clear in general which of the achievable rate in \eqref{eq:DiamonRateFb} or  \eqref{eq:DiamonRateDDF} is larger.

As mentioned in Remark \ref{Mark:Fb2NoFb},  by letting  $R_{\textnormal{Fb},k}=0$ and $\hat{Y}_k=\emptyset$,  for all $k\in[2:N]$ in  Theorem \ref{Them:unicast2.1}, we obtain a new lower bound  for the 4-node diamond network \emph{without} feedback. This lower bound in essence is achieved by letting the two relays use partial decode forward to decode different parts of the source message.

\begin{Corollary}
For 4-node DM diamond network $P_{Y_2Y_3|X_1}(y_2,y_3|x_1)P_{Y_4|X_2X_3}(y_4|x_2,x_3)$, any rate $R>0$ is achievable if it satisfies
\begin{subequations}
\begin{IEEEeqnarray}{rCl}
R &\leq& I(X_1,X_2,X_3; Y_4) \\
R &\leq& \min_{r\in\{2,3\} }I(U_0;Y_r|V_0,X_r)+I(U_2;Y_2|V_0,U_0,X_2)\nonumber\\
&& +I(X_1,X_2;Y_4|V_0,U_0,X_2,U_2)\\
R &\leq& \min_{r\in\{2,3\} }I(U_0;Y_r|V_0,X_r)+I(U_3;Y_3|V_0,U_0,X_3)\nonumber\\
&& +I(X_1,X_2;Y_4|V_0,U_0,X_3,U_3)\\
R &\leq& \min_{r\in\{2,3\} }I(U_0;Y_r|V_0,X_r)\nonumber\\
&&+  I(X_1,X_2,U_2,U_3;Y_4|V_0,U_0)\\
R &\leq&\min_{r\in\{2,3\} }I(U_0;Y_r|V_0,X_r)\\
&& +I(U_2;Y_2|U_0,V_0,X_2)+I(U_3;Y_3|U_0,V_0,X_3)\nonumber\\ &&+ I(X_1; \hat{Y}_2,\hat{Y}_3,Y_4|V_0,U_0,X_2,X_3,U_2,U_3) \label{eq:Diff1}
\end{IEEEeqnarray}
\end{subequations}
 for  some pmf 
\begin{IEEEeqnarray*}{rCl}
&& P_{U_0V_0}P_{X_1|V_0U_0X_2X_3U_2U_3}\left[\prod_{r\in\{2,3\}} P_{X_r|V_0}P_{U_r|V_0U_0X_r}\right]. 
\end{IEEEeqnarray*} 
\end{Corollary}


\section{Achievable rates for  DM multiple-relay channels with partial feedback}\label{sec:mRelays}

\subsection{Scheme 1A}\label{sec:scheme1A}

Define 
\begin{IEEEeqnarray*}{rCl}
&&\textbf{l}_{b}:=(l_{2,b},\ldots,l_{N,b})\\
&& \hat{\textbf{l}}_{b}:=(\hat{l}_{2,b},\ldots,\hat{l}_{N,b})
\end{IEEEeqnarray*} for $b\in[1:B+1]$. Let  ${\textbf{l}}_{0}=\textbf{1}_{[N\!-\!1]}$ and $m_{B+1}=1$.

We present a block-Markov coding scheme where a sequence of $B$ i.i.d message $m_b\in[1:2^{nR}]$, $b\in[1:B]$ is sent over $B+1$ blocks. In each block $b\in[1:B+1]$: 
\begin{itemize}
\item  After obtaining all feedback messages $\textbf{l}_{b-1}$, the transmitter sends inputs $x_{1,b}^n(m_{b}|\textbf{l}_{b-1})$.
\item Node  $k\in[2:N]$   uses compress-forward  to compress its observed  outputs $y^n_{k,b}$,   and then forwards the compression message $l_{k,b}$ through feedback pipe and sends the channel inputs $x^n_{k,b+1}(l_{k,b})$ in block $b+1$.
\item  Each Receiver $d\in\set{D}$ uses joint backward decoding to decode  source message $m_b$ and  compression messages $\textbf{l}_{b-1}$. 
\end{itemize}
Note that the transmitter here simply forwards the feedback messages $\textbf{l}_{b-1}$ and can reconstruct  Node $k$'s inputs $x^n_{k,b}(l_{k,b-1})$, for all $k\in[2:N]$. Thus we are able to superpose the transmitter's inputs  $x^n_1$ on the receivers' and relays' inputs $(x^n_2,\ldots,x^n_N)$, which attains cooperation between the transmitter and the receivers\&relays.  

The coding is explained with the help of Table  \ref{tab:CF}.

\begin{table*}[ht!]
\begin{center}
\caption{ Coding scheme 1A for  multicast network with  feedback}
\begin{tabular}{>{\bfseries}lccccc}
\toprule
Block &1& 2& $\ldots$ & $B$ &$B+1$ \\
\midrule
$X_1$& $x^n_{1,1}(m_{1}| {\textbf{1}}_{[N-1]})$ & $x^n_{1,2}(m_{2}| {\textbf{l}}_{1})$& $\ldots$  & $x^n_{1,B}(m_{B}| {\textbf{l}}_{B-1})$ & $x^n_{1,B+1}(1| {\textbf{l}}_{B})$\\
$X_k$ &$x^n_{k,1}(1)$ &$x^n_{k,2}(l_{k,1})$ & $\ldots$   &$x^n_{k,B}(l_{k,B-1})$  & $x^n_{k,B+1}(l_{k,B})$\\
$\hat{Y}_k$&$\hat{y}^n_{k,1}(l_{k,1}|1)$ & $\hat{y}^n_{k,2}(l_{k,2}|l_{k,1})$& $\ldots$ & $\hat{y}^n_{k,B}(l_{k,B}|l_{k,B-1})$ & $\hat{y}^n_{k,B+1}(1|l_{k,B})$ \\
\midrule$Y_d$&$\hat{m}_1$& $(\hat{m}_2, \hat{\textbf{l}}_{1})$ & $\ldots$  & $\leftarrow(\hat{m}_B, \hat{\textbf{l}}_{B-1})$ & $\leftarrow\hat{\textbf{l}}_B$\\
\bottomrule
 \label{tab:CF}
\end{tabular}
\end{center}
\end{table*}




\subsubsection{Codebook}
Fix the pmf 
\begin{IEEEeqnarray}{rCl}
 \left[\prod^N_{k=2} P_{X_k}\right]\left[\prod^N_{k=2}P_{\hat{Y}_k|X_kY_k}\right]P_{X_1|X_2^N}P_{Y_1^N|X_1^N}.
\end{IEEEeqnarray}
  For each block $b\in[1:B+1]$ and $k\in[2:N]$, randomly and independently generate $2^{n\hat{R}_k}$ sequences $x_{k,b}^n(l_{k,b-1})\sim \prod^n_{i=1}P_{X_k}(x_{k,b,i})$, $l_{k,b-1}\in[1:2^{n\hat{R}_k}]$.  For each $l_{k,b-1}$, randomly and independently generate $2^{n\hat{R}_k}$ sequences $\hat{y}_{k,b}^n(l_{k,b}|l_{k,b-1})\sim \prod^n_{i=1}P_{\hat{Y}_k|X_k}(\hat{y}_{k,b,i}|x_{k,b,i})$. For each $\textbf{l}_{b-1}$, randomly and independently generate $2^{n{R}}$ sequences $x_{1,b}^n(m_b| \textbf{l}_{b-1})\sim \prod^n_{i=1}P_{X_1|X^N_2}(x_{1,b,i}|x_{2,b,i},\ldots,x_{N,b,i})$, $m_{b}\in[1:2^{n{R}}]$. 


\subsubsection{Source encoding}
  In each block $b\in[1:B+1]$, assume that the transmitter already knows ${\textbf{l}}_{b-1}$ through feedback pipes. It sends $x^n_{1,b}(m_{b}| {\textbf{l}}_{b-1}).$ 

 To ensure that source node perfectly knows  $\textbf{l}_{b-1}$, we have
\begin{IEEEeqnarray}{rCl}\label{eq:fbrateCF}
\hat{R}_k\leq R_{\textnormal{Fb},k}. 
\end{IEEEeqnarray}


\subsubsection{Relay and receiver encoding}
Relays and receivers both use compress-forward.  In each block $b\in[1:B]$, node  $k\in[2:N]$  compresses $y_{k,b}^{n}$ by finding a unique index $l_{k,b}$ such that
\begin{IEEEeqnarray}{rCl}\label{eq:Sch1CF}
\big(x^n_{k,b}(l_{k,b\!-\!1}), \hat{y}^n_{k,b}(l_{k,b}|l_{k,b\!-\!1}),y^n_{k,b}\big) \in\mathcal{T}^n_{\epsilon/2}(P_{X_kY_k\hat{Y}_k}).\nonumber
\end{IEEEeqnarray}
Then, it sends $l_{k,b}$ through the feedback pipe at rate 
\begin{IEEEeqnarray}{rCl}
\hat{R}_k\leq R_{\textnormal{Fb},k}, \quad  \text{for $k\in [2:N]$. }
\end{IEEEeqnarray}
and in block $b+1$ sends $x^n_{k,b+1}(l_{k,b})$.

By the covering lemma \cite{Gamal'book}, this is successful with high probability if
\begin{IEEEeqnarray}{rCl}
\hat{R}_k> I(\hat{Y}_k;Y_k|X_k)+\delta(\epsilon/2),\quad  \text{for $k\in[2:N]$. }
\end{IEEEeqnarray}

\subsubsection{Decoding}
Receiver  $d\in\set{D}$ performs  joint backward decoding.  For each block $b\in[B+1,\ldots ,1]$, 
it looks for $(\hat{m}_b, \hat{\textbf{l}}_{b-1})$  such that \footnote{Receiver $d\in\set{D}$ knows $l_{d,b-1}$ since it generated this index.  Since each Receiver $d$ makes its own estimate of $m_{b}$ and $\textbf{l}_{b-1}$,  the precise notation is $(\hat{m}^{(d)}_b,\hat{\textbf{l}}^{(d)}_{b-1})$. For simplicity, we omit the superscript $(d)$.} 
\begin{IEEEeqnarray}{rCl}\label{eq:Sch1Dec}
\big(  &&{x^n_{1,b}(\hat{m}_{b}| \hat{\textbf{l}}_{b-1})},  x^n_{2,b}(\hat{l}_{2,b-1}),\ldots, x^n_{N,b}(\hat{l}_{N,b-1}),y_{d,b}^n,\nonumber\\
&&~ \hat{y}^n_{2,b}(\hat{l}_{2,b}|\hat{l}_{2,b-1}),\ldots, \hat{y}^n_{N,b}(\hat{l}_{N,b}|\hat{l}_{N,b-1})\big)\in\mathcal{T}^n_\epsilon(P_{X_1^N\hat{Y}_2^NY_d}).\nonumber
\end{IEEEeqnarray} 

By the independence of the codebooks, the Markov lemma \cite{Gamal'book}, packing lemma \cite{Gamal'book} and  induction on backward decoding,  this step is successful with high probability if
\begin{IEEEeqnarray}{rCl}\label{eq:rate3CF}
&&R+\sum_{k\in\set{T}}\hat{R}_k\nonumber\\
&&\quad<I(X_1,X(\mathcal{T});\hat{Y}(\mathcal{T}^c),Y_d|X(\mathcal{T}^c))+\sum_{k\in\mathcal{T}}H(\hat{Y}_k|X_k)\nonumber\\
&&\quad~-H(\hat{Y}(\mathcal{T})|X_1^N,\hat{Y}(\mathcal{T}^c),Y_d)-\delta(\epsilon)
\end{IEEEeqnarray} 
for all $\set{T}\subset [2:N]$ \footnote{Each receiver knows the compression message it generated, and it doesn't  need to decode all compression indices $(l_{2,b-1}, \ldots,l_{N,b-1} )$,  we therefore have $\set{T}\subset [2:N]$, rather than $\set{T}\subseteq [2:N]$.} with $\mathcal{T}^c\cap\set{D}\neq \emptyset$.

Combining (\ref{eq:fbrateCF}--\ref{eq:rate3CF}) and using Fourier-Motzkin elimination \cite{Gamal'book} to eliminate $\hat{R}_2,\ldots,\hat{R}_N$, we obtain  Theorem \ref{Them:unicast1}.

\subsection{Scheme 1B}\label{sec:scheme1B}
In  Scheme 1A above, the relays and receivers  use only compress-forward. In this subsection  we present a  scheme where the relays perform mixed compress-forward and partial decode-forward. 

Define\begin{IEEEeqnarray*}{rCl}
&&\textbf{l}_{b}:=(l_{2,b},\ldots,l_{N,b})\\
&& \hat{\textbf{l}}_{b}:=(\hat{l}_{2,b},\ldots,\hat{l}_{N,b})
\end{IEEEeqnarray*} for $b\in[1:B+1]$ and let  ${\textbf{l}}_{0}=\textbf{1}_{[N\!-\!1]}$.

Transmission takes place in $B+1$ blocks each consisting of $n$ transmissions, where  a sequence of $B$ i.i.d message $m_b\in[1:2^{nR}]$, $b\in[1:B]$ is sent over $B+1$ blocks.  Split the message $m_b$ into $(m'_b,m''_b)$, where $m'_b$ and $m''_b$ are independently and uniformly distributed over the sets $[1:2^{nR'}]$ and $[1:2^{nR^{''}}]$, respectively, where $R', R''\geq 0$ and so that  
\begin{IEEEeqnarray}{rCl}\label{eq:RateStru1B}
R=R'+R''.
\end{IEEEeqnarray}  Let   $m''_{B+1}=m'_{B+1}=1$. In each block $b\in[1:B+1]$: 
\begin{itemize}
\item  After obtaining all feedback message ${\textbf{l}}_{b-1}$, the transmitter   sends  inputs $x^n_{1,b}(m''_{b}|m'_{b}, m'_{b-1}, {\textbf{l}}_{b-1})$;
\item Each Relay  $r\in\set{R}$   first uses partial decode-forward to decode part of the source message, i.e., $m'_{b}$, and then uses compress-forward  to compress its observed  outputs $y^n_{k,b}$.  Finally, it feeds the compression index $l_{k,b}$ back to the transmitter through the feedback pipe and broadcasts $x_{r,b}^n(m'_{b-1}, l_{r,b-1})$ in  block $b+1$.
\item Each Receiver $d\in\set{D}$ first uses compress-forward  to compress its observed  outputs $y^n_{d,b}$.  Then it feeds the compression index $l_{d,b}$ back to the transmitter through the feedback pipe and broadcasts $x_{d,b}^n( l_{d,b-1})$. Finally, it uses joint backward decoding to decode  source message $({m}'_{b-1},m''_b)$ and  compression messages $\textbf{l}_{b-1}$. 
\end{itemize}

Similar to Scheme 1A, the transmitter's inputs $x^n_{1,b}$ are imposed on $x^n_{2,b},\ldots,x^n_{n,b}$ since it can reconstructs $x_{r,b}^n(m'_{b-1}, l_{r,b-1})$ and $x_{d,b}^n( l_{d,b-1})$,  which attains cooperation between the transmitter and the receivers\&relays.  

The coding is explained with the help of Table \ref{tab:PDF}.

\begin{table*}[ht!]
\begin{center}
\caption{ Coding scheme 1B for  multicast network with partial feedback}
\begin{tabular}{>{\bfseries}lcccccc}
\toprule
Block &1& 2& $\ldots$ & $B$ &$B+1$ \\
\midrule
$X_1$& $x^n_{1,1}(m''_{1}|m'_{1}, 1, {\textbf{1}_{[N-1]}})$ & $x^n_{1,2}(m''_{2}|m'_{2},m'_{1}, {\textbf{l}}_{1})$& $\ldots$   & $x^n_{1,B}(m''_{B}|m'_{B}, m'_{B-1}, {\textbf{l}}_{B-1})$ & $x^n_{1,B+1}(1|1, m'_{B}, {\textbf{l}}_{B})$\\
$X_r$ &$x^n_{r,1}(1,1)$ &$x^n_{r,2}({m}_1,l_{r,1})$ & $\ldots$   &$x^n_{r,B}({m}'_{B-1},l_{r,B-1})$  & $x^n_{r,B+1}({m}'_{B},l_{r,B})$\\
$U_r$ &$u^n_{r,1}(m'_1|1,1)$ &$u^n_{r,2}(m'_2|{m}_1,l_{r,1})$ & $\ldots$   &$u^n_{r,B}(m'_B|{m}'_{B-1},l_{r,B-1})$  & $u_{r,B+1}(1|{m}'_{B},l_{r,B})$\\
$\hat{Y}_r$ &$\hat{y}^n_{r,1}(l_{r,1}|1,1)$ &$\hat{y}^n_{r,2}(l_{r,2}|{m}_1,l_{r,1})$ & $\ldots$   &$\hat{y}^n_{r,B}(l_{r,B}|{m}'_{B-1},l_{r,B-1})$  & $\hat{y}^n_{r,B+1}(1|\hat{m}'_{B},l_{r,B})$\\

$X_d$ &$x^n_{d,1}(1)$ &$x^n_{d,2}(l_{d,1})$ & $\ldots$   &$x^n_{d,B}(l_{d,B-1})$  & $x^n_{d,B+1}(l_{d,B})$\\
$\hat{Y}_d$&$\hat{y}^n_{d,1}(l_{d,1}|1)$ & $\hat{y}^n_{d,2}(l_{d,2}|l_{d,1})$& $\ldots$ &  $\hat{y}^n_{d,B}(l_{d,B}|l_{d,B-1})$ & $\hat{y}^n_{d,B+1}(1|l_{d,B})$ \\
\midrule$Y_d$&$
\hat{m}''_1$& $\leftarrow(\hat{m}''_2,\hat{m}'_1, \hat{\textbf{l}}_{1})$ & $\ldots$ & $\leftarrow(\hat{m}''_B, \hat{m}'_{B-1}, \hat{\textbf{l}}_{B-1})$ & $\leftarrow( \hat{m}'_{B}, \hat{\textbf{l}}_{B})$\\
\bottomrule
 \label{tab:PDF}
\end{tabular}
\end{center}
\end{table*}

\subsubsection{Codebook} 
Fix pmf 
\begin{IEEEeqnarray}{rCl}
&&  \left[\prod_{r\in\set{R}} P_{X_rU_r}\!P_{\hat{Y}_r|U_rX_rY_r}\! \right]\left[\prod_{d\in\set{D}}\!P_{X_d} P_{\hat{Y}_d|X_dY_d}\right]\nonumber\\
&&\quad\quad\times  P_{X_1|X^N_2U(\set{R})}P_{Y_1^N|X_1^N}.
\end{IEEEeqnarray}
 For  each $r\in\mathcal{R}$ and block $b\in[1:B+1]$,  randomly and independently generate $2^{n(R'+\hat{R}_r)}$ sequences $x_{r,b}^n(m'_{b-1}, l_{r,b-1})\sim \prod^n_{i=1}P_{X_r}(x_{r,b,i})$, with $m'_{b-1}\in[1:2^{nR'}]$ and $l_{r,b-1}\in[1:2^{n\hat{R}_r}]$.  For each $(m'_{b-1},{l}_{r,b-1})$, randomly and independently generate $2^{nR'}$ sequences $u_{r,b}^n(m'_{b}|m'_{b-1},{l}_{r,b-1})\sim \prod^n_{i=1}P_{U_r|X_r}(u_{r,b,i}|x_{r,b,i})$. For each $(m'_{b},m'_{b-1},l_{r,b-1})$, randomly and independently generate $2^{n\hat{R}_r}$ sequences $\hat{y}_{r,b}^n(l_{r,b}|m'_{b},m'_{b-1},l_{r,b-1})\sim \prod^n_{i=1}P_{\hat{Y}_r|U_rX_r}(\hat{y}_{r,b,i}|u_{r,b,i},x_{r,b,i})$. 
 
   For each $d\in\set{D}$ and block $b\in[1:B+1]$,  randomly and independently generate $2^{n\hat{R}_d}$ sequences $x_{d,b}^n(l_{d,b-1})\sim \prod^n_{i=1}P_{X_d}(x_{d,b,i})$,  $l_{d,b-1}\in[1:2^{n\hat{R}_d}]$. For each $l_{d,b-1}$, randomly and independently generate $2^{n\hat{R}_d}$ sequences $\hat{y}_{d,b}^n(l_{d,b}|l_{d,b-1})\sim \prod^n_{i=1}P_{\hat{Y}_d|X_d}(\hat{y}_{d,b,i}|x_{d,b,i})$. 

For each $(m'_{b},m'_{b-1}, \textbf{l}_{b-1})$, randomly and independently generate $2^{n{R''}}$ sequences $x_{1,b}^n(m''_{b}|m'_{b}, m'_{b-1}, \textbf{l}_{b-1})\sim \prod^n_{i=1}P_{X_1|U(\set{R})X^N_2}(x_{1,b,i}|\{u_{r,b,i}\}_{r\in\set{R}},x_{2,b,i},\ldots,x_{N,b,i})$. 




\subsubsection{Source  encoding}
 In each block $b\in[1:B+1]$, assume that the transmitter already knows ${\textbf{l}}_{b-1}$ through the feedback pipes. It sends $
x^n_{1,b}(m''_{b}|m'_{b},m'_{b-1}, \textbf{l}_{b-1})$.

 To ensure that the transmitter perfectly knows  $\textbf{l}_{b-1}$, we have
\begin{IEEEeqnarray}{rCl}\label{eq:fbratePDF2}
\hat{R}_k\leq R_{\textnormal{Fb},k},\quad  \text{for $k\in [2:N]$. }
\end{IEEEeqnarray}

\subsubsection{Relay  encoding}
Relay  nodes perform the mixed compress-forward and decode-forward.   For each block $b\in[1:B+1]$, assume that Relay $r\in\set{R}$ already knows  $\hat{m}'_{b-1}$ from block $b-1$. It looks for a unique index $\hat{m}'_{b}$ such that \footnote{Since each Relay $r\!\in\!\set{R}$   makes its own estimate of $m'_{b}$,  the precise notation  $\hat{m}'^{(r)}_b$. For simplicity, we omit the superscript $(r)$.} 
\begin{IEEEeqnarray}{rCl}\label{eq:relayCom}
\big(  x^n_{r,b}({\hat{m}}'_{b\!-\!1},{l}_{r,b\!-\!1}),\! u^n_{r,b}(\hat{m}'_{b}|{\hat{m}}'_{b\!-\!1},{{l}}_{r,b\!-\!1}), \!y^n_{r,b}\big) \in\mathcal{T}^n_{\epsilon/4}(P_{X_rY_rU_r}).\nonumber
\end{IEEEeqnarray}
then it compresses $y_{r,b}^{n}$ by finding a unique index $l_{r,b}$ such that
\begin{IEEEeqnarray}{rCl}\label{eq:relayCpres}
\big(&&u^n_{r,b}(\hat{m}'_{b}|\hat{m}'_{b\!-\!1},{{l}}_{r,b\!-\!1}),  x^n_{r,b}(\hat{m}'_{b\!-\!1},{l}_{r,b\!-\!1}),\nonumber\\
&&\quad \hat{y}^n_{r,b}(l_{r,b}|\hat{m}'_{b},\!\hat{m}'_{b\!-\!1},{l}_{r,b\!-\!1}),\!y^n_{r,b}\big) \in\mathcal{T}^n_{\epsilon/2}(P_{U_rX_rY_r\hat{Y}_r}).\nonumber
\end{IEEEeqnarray}
Then,  it sends $l_{r,b}$ through the feedback pipe at rate $\hat{R}_r\leq R_{\textnormal{Fb},r}$ and in block $b+1$ sends $x^n_{r,b+1}(\hat{m}'_{b},l_{r,b})$.

By the covering and packing lemmas, this is successful with high probability if for $r\in\set{R}$,
\begin{subequations} \label{eq:rate1PDF}
\begin{IEEEeqnarray}{rCl}
R'&<& I(U_r;Y_r|X_r)-\delta(\epsilon/4)\\
\hat{R}_r&>&I(\hat{Y}_r;Y_r|X_r,U_r)+\delta(\epsilon/2).
\end{IEEEeqnarray}
\end{subequations}
\subsubsection{ Receiver  encoding}
Receiver $d\in\set{D}$ compresses $y_{d,b}^{n}$ by finding a unique index $l_{d,b}$ such that
\begin{IEEEeqnarray}{rCl}\label{eq:RxCpress}
\big(x^n_{d,b}(l_{d,b-1}), \hat{y}^n_{d,b}(l_{d,b}|l_{d,b-1}),y^n_{d,b}\big) \in\mathcal{T}^n_{\epsilon/2}(P_{X_dY_d\hat{Y}_d}).\nonumber
\end{IEEEeqnarray}
Then,  it sends $l_{d,b}$ through the feedback pipe at rate \[\hat{R}_d\leq R_{\textnormal{Fb},d}\]  and in block $b+1$ sends 
$
x^n_{d,b+1}(l_{d,b})
$.

By the covering lemma, this is successful with high probability if
\begin{IEEEeqnarray}{rCl}\label{eq:rate2PDF}
\hat{R}_d&>&I(\hat{Y}_d;Y_d|X_d)+\delta(\epsilon/2). 
\end{IEEEeqnarray}

\subsubsection{Decoding}
Receiver  $d\in\mathcal{D}$ performs backward decoding. For each block $b\in[B+1,\ldots ,1]$, it  looks for $(\hat{m}''_b,\hat{m}'_{b-1},\hat{\textbf{l}}_{b-1})$ such that \footnote{Receiver $d\in\set{D}$ knows $l_{d,b-1}$ since it generated this index.  Since each Receiver $d$   makes its own estimate of $({m}''_b,m'_{b-1},\textbf{l}_{b-1})$,  the precise notation is $(\hat{m}''^{(d)}_b,\hat{m}'^{(d)}_{b-1},\hat{\textbf{l}}^{(d)}_{b-1})$. For simplicity, we omit the superscript $(d)$.}
\begin{IEEEeqnarray}{rCl}\label{eq:RxCom}
\big( x^n_{1,b}(\hat{m}''_{b}&& |\hat{m}'_{b},\hat{m}'_{b},\hat{\textbf{l}}_{b-1}), \textbf{x}^n_{b}(\set{R}),\textbf{x}^n_{b}(\set{D}), \textbf{u}^n_{b}(\set{R}),\nonumber\\
&&\quad \hat{\textbf{y}}^n_b(\set{R}), \hat{\textbf{y}}^n_b(\set{D}),y^n_{d,b}\big) \in\mathcal{T}^n_\epsilon(P_{X_1^NU(\set{R})\hat{Y}_2^NY_d})\nonumber
\end{IEEEeqnarray}
where 
\begin{IEEEeqnarray*}{rCl}
&&\textbf{x}^n_{b}(\set{R}):=\{{x}^n_{r,b}(\hat{m}'_{b-1},\hat{l}_{r,b-1}): r\in\set{R}\}\\
&&\textbf{x}^n_{b}(\set{D}):=\{ x^n_{d,b}(\hat{l}_{d,b-1}) : d\in\set{D}\}\\
&&\textbf{u}^n_{b}(\set{R}):=\{u^n_{r,b}(\hat{m}'_{b}|\hat{m}'_{b-1},{\hat{l}}_{r,b-1}) : r\in\set{R}\}\\
&&\hat{\textbf{y}}^n_b(\set{R}):= \{\hat{y}^n_{r,b}(\hat{l}_{r,b}|\hat{m}'_{b},\hat{m}'_{b-1}, \hat{l}_{r,b-1}): r\in\set{R}\}\\
&&\hat{\textbf{y}}^n_b(\set{D}):= \{\hat{y}^n_{d,b}(\hat{l}_{d,b}|\hat{l}_{d,b-1}): d\in\set{D}\}.
\end{IEEEeqnarray*}

 
By the independence of the codebooks, the Markov lemma, packing lemma and the induction on backward decoding,  the decoding is successful with high probability if
 \begin{IEEEeqnarray}{rCl}\label{eq:sumhat1}
 &&R+\hat{R}(\set{T}\cup\set{R})\nonumber\\
&&\quad< I(X_1,X(\set{T}\cup\set{R}),U(\set{R});\hat{Y}(\set{T}^c\cap\set{D}),Y_d|X(\set{T}^c\cap\set{D}))\nonumber\\
&&\quad\quad+ \sum_{k\in\set{R}} H(\hat{Y}_k|U_k,X_k) + \sum_{j\in\set{D}\cap\set{T}} H(\hat{Y}_j|X_j)  \nonumber\\
&&\quad\quad- H(\hat{Y}(\set{T}\cup\set{R})|X_2^N,U_2^N,\hat{Y}(\set{T}^c\cap\set{D}),Y_d)-\delta(\epsilon)
\end{IEEEeqnarray}
and
  \begin{IEEEeqnarray}{rCl}\label{eq:sumhat2}
R''&+&\hat{R}(\set{T})\nonumber\\
&<&  I(X_1,X(\set{T}),U(\set{T});\hat{Y}(\set{T}^c),Y_d|X(\set{T}^c)))\nonumber\\
&&\quad+ \sum_{k\in\set{R}\cap\set{T}} H(\hat{Y}_k|U_k,X_k) + \sum_{j\in\set{D}\cap\set{T}} H(\hat{Y}_j|X_j)  \nonumber\\
&&\quad- H(\hat{Y}(\set{T})|X_2^N,U_2^N,\hat{Y}(\set{T}^c),Y_d)-\delta(\epsilon)
\end{IEEEeqnarray}
for all $\set{T}\subset [2:N]$ with $\mathcal{T}^c\cap\set{D}\neq \emptyset$, where $\mathcal{T}^c$ is the complement of $T$ in $[2:N]$, and $U_d=\emptyset$, for $d\in\set{D}$.

Combining (\ref{eq:RateStru1B}--\ref{eq:sumhat2}), and  using Fourier-Motzkin elimination to eliminate $R',R'', \hat{R}_2,\ldots,\hat{R}_N$, we obtain Theorem \ref{Them:unicast2}. 

\subsection{Scheme 1C}\label{sec:scheme1C}
{In Scheme 1B, all relay nodes decode the same part of the source message, which may lead to low achievable rates if some relay's observed outputs are very bad. In this subsection, we describe another scheme (Scheme 1C), which  allows different relays to decode different parts of the source message. This new scheme  can  achieve higher rates than Scheme 1A and 1B.}

Transmission takes place in $B+1$ blocks each consisting of $n$ transmissions, where  a sequence of $B$ i.i.d message $m_b\in[1:2^{nR}]$, $b\in[1:B]$ is sent over $B+1$ blocks.  Split the message $m_b$ into \[\big(m'_{0,b},\{m'_{r,b}: r\in\set{R}\}, m''_b\big).\] where  
  messages $m'_{0,b}$, $m'_{r,b}$ and $m''_b$ are independently and uniformly distributed over the sets $[1:2^{nR'_0}]$, $[1:2^{nR'_r}]$ and $[1:2^{nR''}]$, respectively, where $R'_0, R'_r, R''\geq 0$ and so that  
    \begin{IEEEeqnarray}{rCl}\label{eq:RateStru1C}
  R=R'_0+\sum_{r\in\set{R}}R'_r+R''.
  \end{IEEEeqnarray}

\begin{table*}[ht!]
\begin{center}
\caption{ Coding scheme 1C for  multicast network with  feedback}
\begin{tabular}{>{\bfseries}lcccccc}
\toprule
Block &1& $\ldots$ & $B$ &$B+1$ \\
\midrule
$X_1$& $x^n_{1,1}({m}''_{1}|\textbf{m}'_{1}, \textbf{1}_{[|\set{R}|+1]}, {\textbf{1}}_{[N-1]})$ & $\ldots$   & $x^n_{1,B}(m''_{B}|\textbf{m}'_{B}, \textbf{m}'_{B-1}, {\textbf{l}}_{B-1})$ & $x^n_{1,B+1}(1|\textbf{1}_{[|\set{R}|+1]}, \textbf{m}'_{B}, {\textbf{l}}_{B})$\\
$V_0$ & $v_{0,1}^n(1)$ & $\ldots$& $v_{0,B}^n(m'_{0,B-1})$ &  $v_{0,B}^n(m'_{0,B})$\\
$U_0$ & $u_{0,1}^n(m'_{0,1}|1)$ & $\ldots$ & $u_{0,B}^n(m'_{0,B}|m'_{0,B-1})$ & $u_{0,B+1}^n(m'_{0,B}|m'_{0,B-1})$\\
$X_r$ &$x^n_{r,1}(1,1|1)$ & $\ldots$   &$x^n_{r,B}({m}'_{r,B-1},l_{r,B-1}|{m}'_{0,B-1})$  & $x^n_{r,B+1}({m}'_{r,B},l_{r,B}|{m}'_{0,B})$\\
$U_r$ &$u^n_{r,1}(m'_{r,1}|m'_{0,1},1,1,1)$  & $\ldots$   &$u^n_{r,B}(m'_{r,B}|{m}'_{0,B},m'_{0,B\!-\!1},m'_{r,B\!-\!1},{l}_{r,B\!-\!1})$  & $u^n_{r,B+1}(1|1,m'_{0,B},m'_{r,B},{l}_{r,B})$\\

$\hat{Y}_r$ &$\hat{y}_{r,1}^n(l_{r,1}|{m}'_{0,1},m'_{r,1},1,1,1)$ & $\ldots$   &$\hat{y}_{r,B}^n(l_{r,B}|{m}'_{0,B},m'_{r,B},m'_{0,B\!-\!1},m'_{r,B\!-\!1},l_{r,B\!-\!1})$  & $\hat{y}_{r,B+1}^n(1|1,m'_{r,B},m'_{0,B},m'_{r,B},l_{r,B})$\\

$X_d$ &$x^n_{d,1}(1)$  & $\ldots$   &$x^n_{d,B}(l_{d,B-1})$  & $x^n_{d,B+1}(l_{d,B})$\\
$\hat{Y}_d$&$\hat{y}^n_{d,1}(l_{d,1}|1)$ &$\ldots$ &  $\hat{y}^n_{d,B}(l_{d,B}|l_{d,B-1})$ & $\hat{y}^n_{d,B+1}(1|l_{d,B})$ \\
\midrule$Y_d$&$
\hat{m}''_1$ & $\ldots$ & $\leftarrow(\hat{m}''_B, \hat{\textbf{m}}'_{B-1}, \hat{\textbf{l}}_{B-1})$ & $\leftarrow( \hat{\textbf{m}}'_{B}, \hat{\textbf{l}}_{B})$\\
\bottomrule
 \label{tab:PDFa}
\end{tabular}
\end{center}
\end{table*}

Define
\begin{IEEEeqnarray*}{rCl}
&&\textbf{m}'_{b}:=\big(m'_{0,b},\{m'_{r,b}: r\in\set{R}\}\big)\\
&&\textbf{l}_{b}:=(l_{2,b},\ldots,l_{N,b})\\
&& \hat{\textbf{l}}_{b}:=(\hat{l}_{2,b},\ldots,\hat{l}_{N,b})\\
&&\hat{\textbf{m}}'_b=\big(\hat{m}'_{0,b},\{\hat{m}'_{r,b}: r\in\set{R}\}\big)\\
\end{IEEEeqnarray*}
for $b\in[1:B+1]$. Let   $m''_{B+1}=m'_{r,B+1}=m'_{r,0}=m'_{0,0}=1$ and ${\textbf{l}}_{0}=\textbf{1}_{[N\!-\!1]}$. In each block $b\in[1:B+1]$:
\begin{itemize}
\item After obtaining all feedback messages $\textbf{l}_{b-1}$, the transmitter sends $x^n_{1,b}(m''_{b}|\textbf{m}'_{b}, \textbf{m}'_{b-1}, {\textbf{l}}_{b-1})$. 
\item Each Relay  $r\in\set{R}$  decodes $(m'_{0,b},m'_{r,b})$ and generates the compression message $l_{r,b}$ by compressing its channel outputs $y^n_{r,b}$. Then, it forwards $l_{r,b}$ to the transmitter over the feedback pipe and sends the channel inputs $x^n_{r,b+1}(m'_{r,b},l_{r,b}|m'_{0,b})$ in block $b+1$. 
\item Each Receiver $d\in\set{D}$ first compresses its channel outputs $y^n_{d,b}$ and then forwards the compression message $l_{d,b}$ through feedback pipe and sends the channel inputs $x^n_{d,b+1}(l_{r,b})$ in block $b+1$. Finally, it uses joint backward decoding to decode  source message $(\textbf{m}'_{b-1},m''_b)$ and  compression messages $\textbf{l}_{b-1}$.
\end{itemize}
Similar to Scheme 1A, the transmitter's inputs $x^n_{1,b}$ are superposed on $(x^n_{2,b},\ldots,x^n_{n,b})$ since it can reconstruct $x_{k,b}^n$, for all  $k\in[2:N]$,  which attains cooperation between the transmitter and the receivers\&relays.  

 The coding is explained with the help of Table \ref{tab:PDFa}.


\subsubsection{Codebook} 
Fix pmfs  
\begin{IEEEeqnarray}{rCl}
&&P_{V_0}P_{U_0|V_0}\left[\prod_{r\in\set{R}} P_{X_r|V_0}P_{U_r|V_0U_0X_r}\!P_{\hat{Y}_r|V_0U_0U_rX_rY_r}\! \right] \nonumber\\
&&\qquad\times\left[\prod_{d\in\set{D}}\! P_{X_d}P_{\hat{Y}_d|X_dY_d}\right]P_{X_1|X^N_2U(\set{R})}.
\end{IEEEeqnarray}
 For  each  block $b\in[1:B+1]$, randomly and independently generate $2^{nR'_0}$ sequences $v_{0,b}^n(m'_{0,b-1})\sim \prod^n_{i=1}P_{V_0}(v_{b,i})$, with $m'_{0,b-1}\in[1:2^{nR'_0}]$. For each $m'_{0,b-1}$, randomly and independently generate $2^{nR'_0}$ sequences $u_{0,b}^n(m'_{0,b}|m'_{0,b-1})\sim \prod^n_{i=1}P_{U_0|V_0}(u_{0,b,i}|v_{0,b,i})$.
  
 For  each $r\in\mathcal{R}$ and each $m'_{0,b-1}$,   randomly and independently generate $2^{n(R_r'+\hat{R}_r)}$ sequences \\$x_{r,b}^n(m'_{r,b-1}, l_{r,b-1}|m'_{0,b-1})\sim \prod^n_{i=1}P_{X_r|V_0}(x_{r,b,i}|v_{0,b,i})$, with $m'_{r,b-1}\in[1:2^{nR'_r}]$ and $l_{r,b-1}\in[1:2^{n\hat{R}_r}]$.  For each $(m'_{r,b-1},{m}'_{0,b}, m'_{0,b-1}, l_{r,b-1}),$~ randomly and independently generate $2^{nR'_r}$ sequences $u_{r,b}^n(m'_{r,b}|{m}'_{0,b}$, $m'_{0,b-1},m'_{r,b-1},{l}_{r,b-1})\sim \prod^n_{i=1}P_{U_r|V_0U_0X_r}(u_{r,b,i}|v_{0,b,i},u_{0,b,i},x_{r,b,i})$. For each $({m}'_{0,b},m'_{0,b-1},m'_{r,b},m'_{r,b-1},l_{r,b-1})$, randomly and independently generate $2^{n\hat{R}_r}$ sequences $\hat{y}_{r,b}^n(l_{r,b}|{m}'_{0,b},m'_{r,b},m'_{0,b-1},m'_{r,b-1},l_{r,b-1})\sim$~\\$\prod^n_{i=1}P_{\hat{Y}_r|V_0U_0U_rX_r}(\hat{y}_{r,b,i}|v_{0,b,i},u_{0,b,i},u_{r,b,i},x_{r,b,i})$. 
 
   For each $d\in\set{D}$ and block $b\in[1:B+1]$,  randomly and independently generate $2^{n\hat{R}_d}$ sequences $x_{d,b}^n(l_{d,b-1})\sim \prod^n_{i=1}P_{X_d}(x_{d,b,i})$,  $l_{d,b-1}\in[1:2^{n\hat{R}_d}]$. For each $l_{d,b-1}$, randomly and independently generate $2^{n\hat{R}_d}$ sequences $\hat{y}_{d,b}^n(l_{d,b}|l_{d,b-1})\sim \prod^n_{i=1}P_{\hat{Y}_d|X_d}(\hat{y}_{d,b,i}|x_{d,b,i})$. 

For each $(\textbf{m}'_{b},\textbf{m}'_{b-1}, \textbf{l}_{b-1})$, randomly and independently generate $2^{n{R''}}$ sequences $x_{1,b}^n(m''_{b}|\textbf{m}'_{b}, \textbf{m}'_{b-1}, \textbf{l}_{b-1})\sim \prod^n_{i=1}P_{X_1|V_0U_0U(\set{R})X^N_2}(x_{1,b,i}|v_{0,b,i},u_{0,b,i},x_{2,b,i},\ldots,x_{N,b,i}$, $\{u_{r,b,i}: r\in\set{R}\})$. 





\subsubsection{Source  encoding}
 In each block $b\in[1:B+1]$, assume that the transmitter already knows ${\textbf{l}}_{b-1}$ through the feedback pipes. It sends $
x^n_{1,b}(m''_{b}|\textbf{m}'_{b},\textbf{m}'_{b-1}, \textbf{l}_{b-1})$.

 To ensure that the transmitter perfectly knows  $\textbf{l}_{b-1}$, we have
\begin{IEEEeqnarray}{rCl}\label{eq:fbratePDF}
\hat{R}_k\leq R_{\textnormal{Fb},k},\quad  \text{for $k\in [2:N]$. }
\end{IEEEeqnarray}

\subsubsection{Relay  encoding}
Relay  nodes perform hybrid compress-forward and decode-forward.   For each block $b\in[1:B+1]$, assume that Relay $r\in\set{R}$ already knows  $(\hat{m}'_{0,b-1},\hat{m}'_{r,b-1})$ from block $b-1$. It looks for a unique index $\hat{m}'_{0,b}$ such that  
\begin{IEEEeqnarray}{rCl}\label{eq:relayCom}
&&\big(v^n_{0,b}(\hat{m}'_{0,b-1}),u^n_{0,b}(\hat{m}'_{0,b}|\hat{m}'_{0,b-1}),  y^n_{r,b},\nonumber\\
&&\quad  x^n_{r,b}(\hat{{m}}'_{r,b-1},{l}_{r,b-1}|\hat{m}'_{0,b-1}) \big) \in\mathcal{T}^n_{\epsilon/8}(P_{V_0U_0X_rY_r}).\nonumber
\end{IEEEeqnarray}
Then  it looks for $\hat{m}'_{r,b}$ such that 
\begin{IEEEeqnarray}{rCl}\label{eq:relayCom}
&&\big(v^n_{0,b}(\hat{m}'_{0,b-1}),u^n_{0,b}(\hat{m}'_{0,b}|\hat{m}'_{0,b-1}),\nonumber\\
&&\quad u^n_{r,b}(\hat{m}'_{r,b}|\hat{m}'_{0,b},\hat{m}'_{0,b-1},{{m}}'_{r,b-1},{{l}}_{r,b-1})\nonumber\\
&&\quad\quad  x^n_{r,b}(\hat{{m}}'_{r,b-1},{l}_{r,b-1}|\hat{m}'_{0,b-1}), y^n_{r,b} \big) \in\mathcal{T}^n_{\epsilon/4}(P_{V_0U_0X_rY_rU_r}).\nonumber
\end{IEEEeqnarray}
After decoding $(\hat{m}'_{0,b},\hat{m}'_{r,b})$, it compresses $y_{r,b}^{n}$ by finding a unique index $l_{r,b}$ such that
\begin{IEEEeqnarray}{rCl}\label{eq:relayCpres}
&&\big(v^n_{0,b}(\hat{m}'_{0,b-1}),u^n_{0,b}(\hat{m}'_{0,b}|\hat{m}'_{0,b-1}),\nonumber\\
&&\quad u^n_{r,b}(\hat{m}'_{r,b}|\hat{m}'_{0,b},\hat{m}'_{0,b-1},{{m}}'_{r,b-1},{{l}}_{r,b-1})\nonumber\\
&&\quad\quad \hat{y}^n_{r,b}(l_{r,b}|\hat{m}'_{0,b},\hat{m}'_{0,b-1},\hat{m}'_{r,b},\!\hat{m}'_{r,b\!-\!1},{l}_{r,b\!-\!1}),y^n_{r,b},\nonumber\\
&& \qquad\quad x^n_{r,b}(\hat{{m}}'_{r,b-1},{l}_{r,b-1}|\hat{m}'_{0,b-1}) \big) \in\mathcal{T}^n_{\epsilon/2}(P_{V_0U_0U_rX_rY_r\hat{Y}_r}).\nonumber
\end{IEEEeqnarray}
Finally,  it sends $l_{r,b}$ through the feedback pipe at rate 
\begin{IEEEeqnarray}{rCl}\label{eq:MNFb3}
\hat{R}_r\leq R_{\textnormal{Fb},r}
\end{IEEEeqnarray}
 and in block $b+1$ it sends $x^n_{r,b+1}(\hat{m}'_{r,b},l_{r,b}|\hat{m}'_{0,b})$.

By the covering and packing lemmas, these are successful with high probability if for $r\in\set{R}$,
\begin{subequations}\label{eq:rate1PDF}
\begin{IEEEeqnarray}{rCl} 
R_0'&<& I(U_0;Y_r|V_0,X_r)-\delta(\epsilon/8)\\
R_r'&<& I(U_r;Y_r|X_r,U_0,V_0)-\delta(\epsilon/4)\\
\hat{R}_r&>&I(\hat{Y}_r;Y_r|U_0,X_r,U_r,V_0)+\delta(\epsilon/2)
\end{IEEEeqnarray}
\end{subequations}

\subsubsection{ Receiver  encoding}
Receiver $d\in\set{D}$ compresses $y_{d,b}^{n}$ by finding a unique index $l_{d,b}$ such that
\begin{IEEEeqnarray}{rCl}\label{eq:RxCpress}
\big(x^n_{d,b}(l_{d,b-1}), \hat{y}^n_{d,b}(l_{d,b}|l_{d,b-1}),y^n_{d,b}\big) \in\mathcal{T}^n_{\epsilon/2}(P_{X_dY_d\hat{Y}_d}).\nonumber
\end{IEEEeqnarray}
Then,  it sends $l_{d,b}$ through the feedback pipe at rate 
\begin{IEEEeqnarray}{rCl}
\hat{R}_d\leq R_{\textnormal{Fb},d}
\end{IEEEeqnarray}
  and in block $b+1$ sends 
$
x^n_{d,b+1}(l_{d,b})
$.

By the covering lemma, this is successful with high probability if
\begin{IEEEeqnarray}{rCl}\label{eq:rate2PDF}
\hat{R}_d&>&I(\hat{Y}_d;Y_d|X_d)+\delta(\epsilon/2),\quad \text{for $d\in\set{D}$}.~
\end{IEEEeqnarray}

\subsubsection{Decoding}
Receiver  $d\in\mathcal{D}$ performs backward decoding. For each block $b\in[B+1,\ldots ,1]$, 
it  looks for $(\hat{m}''_b,\hat{\textbf{m}}'_{b-1},\hat{\textbf{l}}_{b-1})$ such that \footnote{Receiver $d\in\set{D}$ knows $l_{d,b-1}$ since it generated this index.  Since each Receiver $d$   makes its own estimate of $({m}''_b,\textbf{m}'_{b-1},\textbf{l}_{b-1})$,  the precise notation is $(\hat{m}''^{(d)}_b,\hat{\textbf{m}}'^{(d)}_{b-1},\hat{\textbf{l}}^{(d)}_{b-1})$. For simplicity, we omit the superscript $(d)$.}
\begin{IEEEeqnarray}{rCl}\label{eq:RxCom}
\big( x^n_{1,b}&&(\hat{m}''_{b} |\hat{\textbf{m}}'_{b},\hat{\textbf{m}}'_{b-1},\hat{\textbf{l}}_{b-1}), \textbf{x}^n_{b}(\set{R}),\textbf{x}^n_{b}(\set{D}),\nonumber\\
&& v^n_{0,b}(\hat{m}'_{0,b-1}),u^n_{0,b}(\hat{m}'_{0,b}|\hat{m}'_{0,b-1}),\textbf{u}^n_{b}(\set{R}),\nonumber\\
&& \qquad \hat{\textbf{y}}^n_b(\set{R}), \hat{\textbf{y}}^n_b(\set{D}),y^n_{d,b}\big) \in\mathcal{T}^n_\epsilon(P_{V_0U_0X_1^NU(\set{R})\hat{Y}_2^NY_d})\nonumber
\end{IEEEeqnarray}
where 
\begin{IEEEeqnarray*}{rCl}
&& \textbf{x}^n_{b}(\set{R}):=\{{x}^n_{r,b}(\hat{m}'_{r,b-1},\hat{l}_{r,b-1}|\hat{m}'_{0,b-1}): r\in\set{R}\}\\
&&\textbf{x}^n_{b}(\set{D}):=\{ x^n_{d,b}(\hat{l}_{d,b-1}) : d\in\set{D}\} \\
&&\textbf{u}^n_{b}(\set{R}):=\{u^n_{r,b}(\hat{m}'_{r,b}|\hat{m}'_{0,b},\hat{m}'_{0,b-1},\hat{m}'_{r,b-1},{\hat{l}}_{r,b-1}) : r\in\set{R}\}\\
&& \hat{\textbf{y}}^n_b(\set{R}):= \{\hat{y}^n_{r,b}(\hat{l}_{r,b}|\hat{m}'_{0,b},\hat{m}'_{r,b},\hat{m}'_{0,b\!-\!1},\hat{m}'_{r,b\!-\!1}, \hat{l}_{r,b\!-\!1}): r\in\set{R}\}\\
&& \hat{\textbf{y}}^n_b(\set{D}):= \{\hat{y}^n_{d,b}(\hat{l}_{d,b}|\hat{l}_{d,b-1}): d\in\set{D}\}.
\end{IEEEeqnarray*}
 
By the independence of the codebooks, the Markov lemma, packing lemma and the induction on backward decoding,  the decoding is successful with high probability if
 \begin{IEEEeqnarray}{rCl}\label{eq:sumhatTh2.1}
 &&R''+\sum_{k\in\set{T}}\hat{R}_k+\sum_{k\in\set{T}\cap\set{R}}R'_k\nonumber\\
&&\quad< I(X_1,X(\set{T}),U(\set{T});\hat{Y}(\set{T}^c),Y_d|V_0,U_0,X(\set{T}^c),U({\set{T}^c}))\nonumber\\
&&\quad\quad+ \sum_{k\in\set{R}\cap\set{T}} H(\hat{Y}_k|V_0,U_0,U_k,X_k) + \sum_{j\in\set{D}\cap\set{T}} H(\hat{Y}_j|X_j)  \nonumber\\
&&\quad\quad- H(\hat{Y}(\set{T})|V_0,U_0,X_1^N,U_2^N,\hat{Y}(\set{T}^c),Y_d)-\delta(\epsilon)
\end{IEEEeqnarray}
and
 \begin{IEEEeqnarray}{rCl}\label{eq:sumhatTh2.2}
 &&R+\sum_{k\in\set{T}\cup\set{R}}\hat{R}_k\nonumber\\
&&\quad< I(V,U_0,X_1,X(\set{T}\cup\set{R}),U(\set{R});\nonumber\\
&&\hspace{4cm}\hat{Y}(\set{T}^c\cap\set{D}),Y_d|X(\set{T}^c\cap\set{D}))\nonumber\\
&&\quad\quad+ \sum_{k\in\set{R}} H(\hat{Y}_k|V_0,U_0,U_k,X_k) + \sum_{j\in\set{D}\cap\set{T}} H(\hat{Y}_j|X_j)  \nonumber\\
&&\quad\quad- H(\hat{Y}(\set{T}\cup\set{R})|V_0,U_0,X_1^N,U_2^N,\hat{Y}(\set{T}^c\cap\set{D}),Y_d)-\delta(\epsilon)\nonumber\\
\end{IEEEeqnarray}
for all $\set{T}\subset [2:N]$ with $\mathcal{T}^c\cap\set{D}\neq \emptyset$, where $\mathcal{T}^c$ is the complement of $T$ in $[2:N]$, and $U_d=\emptyset$, for $d\in\set{D}$. 

Combining (\ref{eq:RateStru1C}--\ref{eq:sumhatTh2.2}), and  using Fourier-Motzkin elimination to eliminate $\hat{R}_2,\ldots,\hat{R}_N, R''$ and $R'_r$, for $r\in\set{R}$, we obtain Theorem \ref{Them:unicast2.1}.

\section{Discrete Memoryless Multicast Network}\label{Sec:extension}
In Section \ref{sec:mRelays} we proposed   block-Markov coding schemes for  DM-MN with instantaneous,  rate-limited  and noisy-free  feedback. Recall the NNC scheme \cite{Yassaee'08, Lim'11, Hou'13} for  DM-MN without feedback, where each node (including the transmitter)  compresses its observed signals  and sends the corresponding compression message in the next block.  Comparing our  coding schemes with NNC, we observe that both schemes involve block-Markov coding, compressing channel outputs and sending compression messages. However, our scheme allows hybrid relaying strategies at relay nodes, and  in each block,  instead of creating a new compression index, the transmitter forwards all compression indices generated by the receivers and relays. In our scheme, the transmitter and the relays\&receivers cooperate with each other through feedback, and different nodes  operate differently according to the features of the network, which leads to a larger achievable rate than NNC, as shown by examples in Section \ref{sec:Eg}.  


Motivated by our feedback coding scheme, we propose another  scheme  for  DM-MN \emph{without} feedback. The key idea is that in each block, instead of obtaining the compression messages  through feedback pipes, the transmitter decodes them based  on its observed channel outputs. One must be cautious that when in absence of feedback, the transmitter's inputs $x^n_{1,b}$ cannot be superposed on  the receivers' and relays' inputs $(x^n_{2,b},\ldots,x^n_{N,b})$ like the feedback case. This is because, in each block $b$, Node $k$  creates compression message $l_{k,b}$ by compressing compressing $y^n_{k,b}$, and  sends $x_{k,b+1}(\cdot,l_{k,b})$ in block $b+1$. The transmitter has to wait to observe the channel outputs $y^n_{1,b+1}$ and then decodes $l_{k,b}$,  which means that at the beginning of each block $b$, the transmitter can only reconstruct $x^n_{k,b-1}$ before the transmission. 

To ensure the cooperation between the transmitter and the receivers\&relays, we made the following modification:  Transmission takes place in $B+2$ blocks each consisting of $n$ transmissions. In each block $b$, each Node $k\in [2:N]$   creates a compression index $l_{k,b-1}$ and sends   $(l_{k,b-1},l_{k,b-2})$. The transmitter, after observing $y^n_{1,b}$,  first decodes compression indices $\textbf{l}_{b-1}$, which is in essence a coding problem on a multiple access channel $P_{Y_1|X_2,\ldots,X_N}$ with side information $x^n_{1,b}$.  Then in block $b+1$, the transmitter sends compression messages $\textbf{l}_{b-1}$ with source message $m_{b+1}$.  The coding is explained with the help of  \ref{tab:PDFNoisy}.

\begin{table*}[ht!]
\begin{center}
\caption{ Coding scheme  for  multicast network without feedback}
\begin{tabular}{>{\bfseries}lcccccc}
\toprule
Block &1& $\ldots$ & $B$ &$B+1$ &$B+2$ \\
\midrule
$X_1$& $x^n_{1,1}(m''_{1}|m'_{1}, 1, {\textbf{1}})$ 
& $\ldots$   & $x^n_{1,B}(m''_{B}|m'_{B}, m'_{B\!-\!1}, {\textbf{l}}_{B\!-\!2})$ & $x^n_{1,B+1}(1|1, m'_{B}, {\textbf{l}}_{B\!-\!1})$ & $x^n_{1,B+2}(1|1, 1, {\textbf{l}}_{B})$ \\

$V_r$ &$v_{r,1}(1,1)$ 
& $\ldots$ &$v^n_{r,B}({m}'_{B\!-\!1},l_{r,B\!-\!2})$  & $v^n_{r,B+1}({m}'_{B},l_{r,B\!-\!1})$ &$v^n_{r,B+2}(1,l_{r,B})$\\

$X_r$ &$x^n_{r,1}(1|1,1)$ 
& $\ldots$ &$x^n_{r,B}(l_{r,B\!-\!1}|{m}'_{B\!-\!1},l_{r,B\!-\!2})$  & $x^n_{r,B+1}(l_{r,B}|\hat{m}'_{B},l_{r,B\!-\!1})$ &$x^n_{r,B+2}(1|1,l_{r,B})$\\
$U_r$ &$u_{r,1}(m'_1|1,1)$ 
& $\ldots$   &$u^n_{r,B}(m'_B|{m}'_{B\!-\!1},l_{r,B\!-\!2})$  & $u^n_{r,B+1}(1|{m}'_{B},l_{r,B-1})$&$u^n_{r,B+2}(1|1,l_{r,B})$\\
$\hat{Y}_r$ &$\hat{y}^n_{r,1}(l_{r,1}|m'_1,1,1,1)$
 & $\ldots$   &$\hat{y}^n_{r,B}(l_{r,B}|{m}'_{B\!-\!1},m'_B,l_{r,B\!-\!2},l_{r,B\!-\!1})$  & $\hat{y}_{r,B+1}(1|{m}'_{B},1,l_{r,B-1},l_{r,B})$ & $\hat{y}^n_{r,B+2}(1|1,1,l_{r,B},1)$\\

$V_d$ &$v^n_{d,1}(1)$
 & $\ldots$   &$v^n_{d,B}(l_{d,B\!-\!2})$  & $v^n_{d,B+1}(l_{d,B\!-\!1})$ &$v^n_{d,B+2}(l_{d,B})$\\
$X_d$ &$x^n_{d,1}(1|1)$
 & $\ldots$   &$x^n_{d,B}(l_{d,B\!-\!1}|l_{d,B\!-\!2})$  & $x^n_{d,B+1}(l_{d,B}|l_{d,B\!-\!1})$&$x^n_{d,B+2}(1|l_{d,B})$\\
$\hat{Y}_d$&$\hat{y}^n_{d,1}(l_{d,1}|1,1)$ 
& $\ldots$ &  $\hat{y}^n_{d,B}(l_{d,B}|l_{d,B\!-\!2},l_{d,B\!-\!1})$ & $\hat{y}^n_{d,B+1}(1|l_{d,B\!-\!1},l_{d,B})$&  $\hat{y}^n_{d,B+2}(1|l_{d,B},1)$ \\
\midrule
$Y_d$&$
\hat{m}''_1$& $\ldots$ & $\leftarrow(\hat{m}''_B, \hat{m}'_{B\!-\!1}, \hat{\textbf{l}}_{B\!-\!2})$ & $\leftarrow( \hat{m}'_{B}, \hat{\textbf{l}}_{B-1})$ & $\leftarrow \hat{\textbf{l}}_{B}$\\
\bottomrule
 \label{tab:PDFNoisy}
\end{tabular}
\end{center}
\end{table*}


\textit{1) Codebook:} 
Fix the pmf \begin{IEEEeqnarray}{rCl}
&& \left[\prod_{k=2}^N P_{V_k}P_{X_k|V_k}P_{U_k|V_k} \right]\left[\prod_{r\in\set{R}}\!  P_{\hat{Y}_r|U_rV_rX_rY_r}\right]\nonumber\\
&&\quad~\times \left[\prod_{d\in\set{D}} \!P_{\hat{Y}_d|V_dX_dY_d}\right]P_{X_1|V^N_2U(\set{R})}P_{Y_1^N|X_1^N}.
\end{IEEEeqnarray}  For block $b\in[1:B]$, split the message $m_b\in[1:2^{nR}]$ into $(m'_b,m''_b)$, where $m'_b$ and $m''_b$ are independently and uniformly distributed over the sets $\in[1:2^{nR'}]$ and $[1:2^{nR^{''}}]$, respectively, where $R', R''\geq 0$ and so that  
\begin{IEEEeqnarray}{rCl}\label{eq:RateStru2}
R=R'+R''.
\end{IEEEeqnarray} 
  Let  ${\textbf{l}}_{-1}=\textbf{l}_{0}=\textbf{1}_{[N\!-\!1]}$ and $m''_{B+1}=m'_{B+1}=m''_{B+2}=m'_{B+2}=1$.

 For  each $r\in\mathcal{R}$ and block $b\in[1:B+2]$,  randomly and independently generate $2^{n(R'+\hat{R}_r)}$ sequences $v_{r,b}^n(m'_{b-1}, l_{r,b-2})\sim \prod^n_{i=1}P_{V_r}(v_{r,b,i})$, with $m'_{b-1}\in[1:2^{nR'}]$ and $l_{r,b-2}\in[1:2^{n\hat{R}_r}]$.  For each $(m'_{b-1},{l}_{r,b-2})$, randomly and independently generate $2^{n\hat{R}_r}$ sequences $x_{r,b}^n(l_{r,b-1}|m'_{b-1}, l_{r,b-2})\sim \prod^n_{i=1}P_{X_r|V_r}(x_{r,b,i}|v_{r,b,i})$.  For each pair $(m'_{b-1},{l}_{r,b-2})$, randomly and independently generate $2^{nR'}$ sequences $u_{r,b}^n(m'_{b}|m'_{b-1},{l}_{r,b-2})\sim \prod^n_{i=1}P_{U_r|V_r}(u_{r,b,i}|v_{r,b,i})$. For each $(m'_{b},m'_{b-1},l_{r,b-2},l_{r,b-1})$, randomly and independently generate $2^{n\hat{R}_r}$ sequences $\hat{y}_{r,b}^n(l_{r,b}|m'_{b},m'_{b-1},l_{r,b-2},l_{r,b-1})\sim \prod^n_{i=1}P_{\hat{Y}_r|U_rX_rV_r}(\hat{y}_{r,b,i}|u_{r,b,i},x_{r,b,i},v_{r,b,i})$. 
 
  For each $d\in\set{D}$ and block $b\in[1:B+2]$,  randomly and independently generate $2^{n\hat{R}_d}$ sequences $v_{d,b}^n(l_{d,b-2})\sim \prod^n_{i=1}P_{V_d}(v_{d,b,i})$,  with $l_{d,b-2}\in[1:2^{n\hat{R}_d}]$. For each $l_{d,b-2}$,  randomly and independently generate $2^{n\hat{R}_d}$ sequences $x_{d,b}^n(l_{d,b-1}|l_{d,b-2})\sim \prod^n_{i=1}P_{X_d|V_d}(x_{d,b,i}|v_{d,b,i})$.  For each $(l_{d,b-2},l_{d,b-1})$, randomly and independently generate $2^{n\hat{R}_d}$ sequences $\hat{y}_{d,b}^n(l_{d,b}|l_{d,b-2},l_{d,b-1})\sim \prod^n_{i=1}P_{\hat{Y}_d|X_dV_d}(\hat{y}_{d,b,i}|x_{d,b,i},v_{d,b,i})$. 

For each $(m'_{b},m'_{b-1}, \textbf{l}_{b-2})$, randomly and independently generate $2^{n{R''}}$ sequences $x_{1,b}^n(m''_{b}|m'_{b}, m'_{b-1}, \textbf{l}_{b-2})\sim \prod^n_{i=1}P_{X_1|U(\set{R})V^N_2}(x_{1,b,i}|v_{2,b,i},\ldots,v_{N,b,i},\{u_{r,b,i}: r\in\set{R}\})$.




Let \begin{IEEEeqnarray*}{rCl}
&&\textbf{v}'^n_{b}(\set{R}):=\{v^n_{r,b}(\hat{m}'_{b-1},\hat{l}_{r,b-2}): r\in\set{R}\}\\
&&\textbf{v}'^n_{b}(\set{D}):=\{v^n_{d,b}(\hat{l}_{d,b-2}): d\in\set{D}\}\\
&&\textbf{x}'^n_{b}(\set{R}):=\{{x}^n_{r,b}(\hat{l}_{r,b-1}|\hat{m}'_{b-1},\hat{l}_{r,b-2}): r\in\set{R}\}\\
&&\textbf{x}'^n_{b}(\set{D}):=\{ x^n_{d,b}(\hat{l}_{d,b-1}|\hat{l}_{d,b-2}) : d\in\set{D}\}\\
&&\textbf{u}'^n_{b}(\set{R}):=\{u^n_{r,b}(\hat{m}'_{b}|\hat{m}'_{b-1},{\hat{l}}_{r,b-2}): r\in\set{R}\}\\
&&\hat{\textbf{y}}'^n_b(\set{R}):= \{\hat{y}^n_{r,b}(\hat{l}_{r,b}|\hat{m}'_{b},\hat{m}'_{b-1}, \hat{l}_{r,b-2},\hat{l}_{r,b-1}): r\in\set{R}\}\\
&&\hat{\textbf{y}}'^n_b(\set{D}):= \{\hat{y}^n_{d,b}(\hat{l}_{d,b}|\hat{l}_{d,b-2},\hat{l}_{d,b-1}): d\in\set{D}\}.
\end{IEEEeqnarray*}


\textit{2) Source  encoding:}
  At each block $b\in[1:B+1]$, after observing $y^n_{1,b}$, it looks for $\hat{\textbf{l}}_{b-1}$ such that
\begin{IEEEeqnarray}{rCl}\label{eq:MNDecodeCF}
 \big(&& x^n_{1,b}({m}''_{b}|{m}'_{b},{m}'_{b-1},\hat{\textbf{l}}_{b-2}),  \textbf{v}'^n_{b}(\set{R}), \textbf{v}'^n_{b}(\set{D}), \textbf{x}'^n_{b}(\set{R}),\nonumber\\&&\quad \textbf{x}'^n_{b}(\set{D}),\textbf{u}'^n_{b}(\set{R}),  y^n_{1,b}\big) \!\in\!\mathcal{T}^n_{\epsilon/8}(P_{V_2^NX_1^NU_2^NY_1})\nonumber
\end{IEEEeqnarray}
where $\hat{m}''_{b}={m}''_{b}$, $\hat{m}'_{b}={m}'_{b}$ and $\hat{m}'_{b-1}={m}'_{b-1}$, since the transmitter knows the source message it sent.


After finding  compression indices $\hat{\textbf{l}}_{b-1}$,  in block $b+1$ the transmitter sends $
x^n_{1,b+1}({m}''_{b+1}|{m}'_{b+1},{m}'_{b},\hat{\textbf{l}}_{b-1})
$.

By the  packing lemma, this step is successful with high probability if for all  subset $\set{J}\subseteq[2:N]$,
\begin{IEEEeqnarray}{rCl}\label{eq:MNfbratePDF}
\hat{R}(\set{J})&<& I(X(\set{J}); Y_1|X(\set{J}^c),V_2^N,U(\set{R}),X_1)-\delta(\epsilon/8).\quad
\end{IEEEeqnarray}

\textit{3) Relay  encoding:}
Relay  nodes perform mixed compress-forward and partial decode-forward.   In each block $b\in[1:B+1]$, assume Relay $r\in\set{R}$ already knows $\hat{m}'_{b-1}$ from previous block. It looks for a unique index $\hat{m}'_{b}$ such that \footnote{Since each Relay $r\in\set{R}$   makes its own estimate of $m'_{b}$,  the precise notation should be $\hat{m}'^{(r)}_b$. For simplicity, we omit the superscript $(r)$.} 
\begin{IEEEeqnarray}{rCl}\label{eq:MNCom}
\big(&&v^n_{r,b}(\hat{m}'_{b-1},l_{r,b-2}),  x^n_{r,b}({l}_{r,b-1}|{\hat{m}}'_{b-1},{l}_{r,b-2}), \nonumber\\&&\quad u^n_{r,b}(\hat{m}'_{b}|{\hat{m}}'_{b-1},{{l}}_{r,b-2}), y^n_{r,b}\big) \in\mathcal{T}^n_{\epsilon/6}(P_{X_rY_rU_rV_r}),\nonumber
\end{IEEEeqnarray}
then it compresses $y_{r,b}^{n}$ by finding a unique index $l_{r,b}$ such that
\begin{IEEEeqnarray}{rCl}\label{eq:MNCpres}
\big(&&v^n_{r,b},u^n_{r,b},  x^n_{r,b},y^n_{r,b},\nonumber\\
&&\quad \hat{y}^n_{r,b}(l_{r,b}|\hat{m}'_{b},\hat{m}'_{b-1},{l}_{r,b-2},{l}_{r,b-1})\big) \in\mathcal{T}^n_{\epsilon/4}(P_{V_rU_rX_rY_r\hat{Y}_r}).\nonumber
\end{IEEEeqnarray}
Then, in block $b+1$ it sends $
x^n_{r,b+1}(l_{r,b}|\hat{m}'_{b},l_{r,b-1})$.

By the covering and  packing lemma, this step is successful with high probability if for $r\in\set{R}$,
\begin{subequations}\label{eq:MNrate1PDF}
\begin{IEEEeqnarray}{rCl} 
R'&<& I(U_r;Y_r|V_r,X_r)-\delta(\epsilon/6)\\
\hat{R}_r&>&I(\hat{Y}_r;Y_r|V_r,X_r,U_r)+\delta(\epsilon/4).
\end{IEEEeqnarray}
\end{subequations}

\textit{4) Receiver  encoding:}
Receiver $d\in\set{D}$ compresses $y_{d,b}^{n}$ by finding a unique index $l_{d,b}$ such that
\begin{IEEEeqnarray}{rCl}\label{eq:MNRxCpress}
\big(&&v^n_{d,b}(l_{d,b-2}), x^n_{d,b}(l_{d,b-1}|l_{d,b-2}), \nonumber\\&&\quad\hat{y}^n_{d,b}(l_{d,b}|l_{d,b-2},l_{d,b-1}),y^n_{d,b}\big) \in\mathcal{T}^n_{\epsilon/4}(P_{V_dX_dY_d\hat{Y}_d}).\nonumber
\end{IEEEeqnarray}
Then, in block $b+1$ it sends $
x^n_{d,b+1}(l_{d,b}|l_{d,b-1})$.

By the covering and  packing lemmas, this step is successful with high probability if
\begin{IEEEeqnarray}{rCl}\label{eq:MNrate2PDF}
\hat{R}_d&>&I(\hat{Y}_d;Y_d|V_d,X_d)+\delta(\epsilon/4),\quad \text{for $d\in\set{D}$}.~~
\end{IEEEeqnarray}

\textit{5) Decoding:}
Receiver  $d\in\mathcal{D}$ performs backward decoding. For each block $b\in[B+2,\ldots ,1]$, it  looks for $(\hat{m}''_b,\hat{m}'_{b-1},\hat{\textbf{l}}_{b-2})$ such that \footnote{Receiver $d\in\set{D}$ knows $l_{d,b-2}$ since it generated itself.  Since each Receiver $d\in$ makes its own estimate of $({m}''_b,m'_{b-1},\textbf{l}_{b-2})$,  the  precise notation is $(\hat{m}''^{(d)}_b,\hat{m}'^{(d)}_{b-1},\hat{\textbf{l}}^{(d)}_{b-2})$. For simplicity, we omit the superscript $(d)$.} 
\begin{IEEEeqnarray}{rCl}\label{eq:MNDecode}
 \big(&& x^n_{1,b}(\hat{m}''_{b}|\hat{m}'_{b},\hat{m}'_{b-1},\hat{\textbf{l}}_{b-2}),  \textbf{v}'^n_{b}(\set{R}), \textbf{v}'^n_{b}(\set{D}), \textbf{x}'^n_{b}(\set{R}),\textbf{x}'^n_{b}(\set{D}), \nonumber\\&&\quad\textbf{u}'^n_{b}(\set{R}), \hat{\textbf{y}}'^n_b(\set{R}),\!\hat{\textbf{y}}'^n_b(\set{D}),\!y^n_{d,b}\big)\! \in\!\mathcal{T}^n_\epsilon(\!P_{V_2^NX_1^NU(\set{R})\hat{Y}_2^NY_d}\!).\nonumber
\end{IEEEeqnarray}

By the independence of the codebooks, the Markov lemma, packing lemma and  induction on backward decoding,  the decoding is successful with high probability if
 \begin{IEEEeqnarray}{rCl}\label{eq:MNsumhat1}
R&+&\hat{R}(\set{T}\cup\set{R})\nonumber\\
&<&  I(X_1,V(\set{T}\cup\set{R}),\!U(\set{R}),\!X(\set{T}\cup\set{R});\!\hat{Y}(\set{T}^c\cap\set{D}),\!Y_d|\nonumber\\
&&\hspace{3cm}V(\set{T}^c),\!X(\set{T}^c\cap\set{D}))\nonumber\\
&&+ \sum_{k\in\set{R}} H(\hat{Y}_k|X_k,U_k,V_k)+\sum_{j\in\set{D}\cap\set{T}} H(\hat{Y}_j|X_j,V_j)\nonumber\\
&&-H(\hat{Y}(\set{T}\cup\set{R})|V_2^N,X_1^N,U_2^N,\hat{Y}(\set{T}^c\cap\set{D}),Y_d)-\delta(\epsilon)\nonumber\\
\end{IEEEeqnarray}
and
  \begin{IEEEeqnarray}{rCl}\label{eq:MNsumhat2}
R''&+&\hat{R}(\set{T})\nonumber\\
&<&  I(X_1,\!V(\set{T}),\!U(\set{T}),\!X(\set{T});\!\hat{Y}(\set{T}^c),\!Y_d|\nonumber\\
&&\hspace{4cm}
V(\set{T}^c),\!X(\set{T}^c),\!U(\set{T}^c))\nonumber\\
&&+ \sum_{k\in\set{R}\cap\set{T}} H(\hat{Y}_k|X_k,U_k,V_k)+\sum_{j\in\set{D}\cap\set{T}} H(\hat{Y}_j|X_j,V_j)\nonumber\\
&&-H(\hat{Y}(\set{T})|V_2^N,X_1^N,U_2^N,\hat{Y}(\set{T}^c),Y_d)-\delta(\epsilon)
%
\end{IEEEeqnarray}
for all $\set{T}\subset [2:N]$  with $\mathcal{T}^c\cap\set{D}\neq \emptyset$ and $U_d=\emptyset$, for all $d\in\set{D}$.

Combining (\ref{eq:RateStru2}--\ref{eq:MNsumhat2}), and  using Fourier-Motzkin elimination to eliminate $R',R'', \hat{R}_2,\ldots,\hat{R}_N$, we obtain Theorem \ref{Them:unicast3}.

\section{conclusion}

\end{document}